\documentclass[aps,prl,reprint,superscriptaddress,nofootinbib,longbibliography]{revtex4-2}

\usepackage{amsmath,amssymb,bm}
\usepackage[colorlinks=true,linkcolor=blue,citecolor=blue,urlcolor=blue]{hyperref}
\usepackage[utf8]{inputenc}
\usepackage{graphicx}
\usepackage{physics}
\usepackage{amsthm}
\usepackage{xcolor}
\newcommand{\id}{\mathbb{I}}
\newcommand{\val}{\operatorname{val}}
\newcommand{\eps}{\varepsilon}
\newcommand{\dist}{\operatorname{dist}}

\newtheorem{lemma}{Lemma}
\newtheorem{theorem}{Theorem}

\makeatletter
\renewcommand\paragraph{\@startsection{paragraph}{4}%
  {\z@}%
  {1.25ex \@plus1ex \@minus.2ex}%
  {-1em}%
  {\normalfont\bfseries}}
\makeatother

\makeatletter
\newcommand{\manualref}[1]{\def\@currentlabel{#1}}
\makeatother

\begin{document}

\title{Adiabatic Quantum Phase Estimation}

\author{Alexander Schmidhuber}
\email{alexsc@mit.edu}
\affiliation{Center for Theoretical Physics, Massachusetts Institute of Technology}

\author{Seth Lloyd}
\email{slloyd@mit.edu}
\affiliation{Department of Mechanical Engineering, Massachusetts Institute of Technology}

\date{\today}

\begin{abstract}
Quantum phase estimation (QPE) is a central algorithmic primitive which estimates eigenvalues of a Hamiltonian up to precision $\eps$ in Heisenberg-limited time $T=\Theta(1/\eps)$. Standard gate-based implementations of QPE require deep controlled time-evolution circuits and are not native to analog hardware.
Here, we present a simple adiabatic protocol for QPE that 
achieves (up to logarithmic factors) the optimal Heisenberg-limited scaling $T = O\left( \frac{1}{\eps} \log\left(\delta^{-1}\right)\right)$ in both the precision $\eps$ and failure probability $\delta$.
By encoding eigenvalues in the populations of computational basis states rather than complex phases, our approach is naturally robust against certain dephasing errors. The adiabatic protocol only requires the ability to couple a single ancilla qubit to the system Hamiltonian as well as pairwise couplings within the ancilla.  
\end{abstract}

\maketitle

\section{Introduction}
Quantum Phase Estimation (QPE) is a fundamental quantum primitive that  underlies the majority of existing quantum applications, including quantum simulation, integer factorization, quantum chemistry, and linear-systems methods~\cite{Kitaev1995,Cleve1998,NielsenChuang,Shor1994,AbramsLloyd1999,AspuruGuzik2005,HHL2009}.
For a Hamiltonian $H_S$ with eigenpairs $H_S\ket{\alpha}=\lambda_\alpha\ket{\alpha}$ and an arbitrary initial state $\ket{\psi}_\mathcal{S}=\sum_\alpha c_\alpha\ket{\alpha}$, idealized QPE implements the operation
\begin{equation}
\sum_\alpha c_\alpha\ket{\alpha}\ket{0}
\mapsto
\sum_\alpha c_\alpha\ket{\alpha}\ket{\widetilde\lambda_\alpha},
\label{eq:QPEact}
\end{equation}
where $\widetilde\lambda_\alpha$ is an $m$-bit estimate of $\lambda_\alpha$. Gate-based QPE achieves additive accuracy $\eps\simeq2^{-m}$ in time
$T=\Theta(2^m)=\Theta(1/\eps)$, saturating the optimal tradeoff between
time, energy, and precision known as the Heisenberg limit~\cite{GLM2006}.

Standard gate-based implementations of QPE require a Quantum Fourier Transform and deep controlled powers of \(e^{-iH_St}\), typically through product formulas, block-encodings, or qubitization methods~\cite{Lloyd1996,BerryChildsKothari2015,LowChuang2019}. Such digital constructions are powerful but expensive and not native to analog platforms in which \(H_S\) is available as a physical interaction. This motivates an adiabatic implementation of QPE, which avoids Trotterization and is naturally robust to certain control and dephasing errors.

In this paper, we present a conceptually simple adiabatic protocol for QPE. Its time complexity is $T = \widetilde{O} \left(\frac{1}{\eps} \log\frac{1}{\delta}\right)$, which is optimal (up to logarithmic factors) in both the accuracy $\eps$ and the failure probability $\delta$~\cite{mande2023tight}. Our protocol requires the ability to couple the system Hamiltonian to a single ancilla qubit at a time via the interaction $H_S \otimes Z_j$, as well as pairwise $Z_iZ_j$ interactions between ancilla qubits. To obtain logarithmic scaling in the failure probability $\delta$, we devise an improved two-level adiabatic schedule that exponentially suppresses diabatic leakage. This schedule also naturally improves the adiabatic version of Grover's algorithm \cite{RolandCerf2002}, and we expect it to apply to a wider class of adiabatic applications.

\paragraph{Challenges.}
A core difficulty when implementing QPE adiabatically is retaining the Heisenberg limit. While an adiabatic implementation follows immediately from the universality proof of adiabatic quantum computation~\cite{Aharonov2007}, such mappings are far from optimal due to unnatural history-state constructions. More broadly, adiabatic critical metrology is known to fall short of Heisenberg scaling in certain settings~\cite{gietka2021adiabatic}. 

Standard adiabatic guarantees yield suboptimal scaling in both the accuracy $\eps$ and the failure probability $\delta$.
Consider a time-varying two-level Hamiltonian $h(t)$ with instantaneous ground and excited state $\ket{g(t)}$ and $\ket{e(t)}$, and associated energies $E_g(t)$ and $E_e(t)$. 
A commonly employed\footnote{While this simple criterion is frequently used and often sets the right scale, it is neither necessary nor sufficient (see e.g. \cite{lidar2009adiabatic,farhi2000quantum,farhi2001quantum,RolandCerf2002,du2008experimental} and our later discussion).} version of the ``adiabatic approximation theorem'' states that diabatic leakage is suppressed provided the Hamiltonian evolves slowly relative to the square of its spectral gap~\cite{messiah1962quantum}. More precisely, the final state $\ket{\psi(T)}$ after adiabatic evolution is said to satisfy $|\bra{\psi(T)}\ket{g(T)}|^2 \geq 1 - \delta$ provided
\begin{equation}
 \forall t: \  \frac{\left| \bra{e(t)}\partial_t h(t)\ket{g(t)}\right|}{(E_e(t)-E_g(t))^2} \leq \sqrt{\delta}.
\label{eq:adiabatic_simple_intro}
\end{equation}
This condition results in two immediate issues. First, applying  the bound \eqref{eq:adiabatic_simple_intro} to a Hamiltonian resolving energies up to precision $\eps$---which must necessarily have a spectral gap of order $O(\eps)$---yields a suboptimal runtime of $T = \Omega(\frac{1}{\eps^2})$ for fixed evolution speed $\partial_t h(t)$, which is quadratically worse than the Heisenberg limit. Generic superadiabatic guarantees require even longer runtimes scaling with the cube of the inverse gap~\cite{lidar2009adiabatic}. Second, the diabatic leakage (and thus the overall failure probability of the algorithm) decays only polynomially in the total interrogation time $T_\mathrm{tot}$ according to \eqref{eq:adiabatic_simple_intro}: Lowering $\delta$ by a factor of four requires halving the evolution speed and thus doubling the total interrogation time. In contrast, the failure probability of gate-based QPE decays exponentially in the number of gates.

Our algorithm resolves the above two obstacles by (i) using a locally varying schedule in the spirit of~\cite{RolandCerf2002} to recover the Heisenberg-optimal scaling in $\eps$ and (ii) reducing the analysis of adiabatic QPE to a two-level system, for which we 
design an adiabatic schedule based on the theory of superadiabaticity~\cite{Nenciu1993,HagedornJoye2002,JRS2007,lidar2009adiabatic}. The schedule is obtained by  reparametrizing the Roland--Cerf local-gap schedule with an endpoint-flat Denjoy--Carleman clock~\cite{RolandCerf2002,Rainer2022,KrieglMichorRainer2009}. 
The use of ultradifferentiable regularity in adiabatic theory has important
precedents: Nenciu proved exponential adiabatic estimates under Gevrey-type
regularity, and Jung and Elgart--Hagedorn developed Gevrey switching theorems
for compactly supported smooth interpolations
\cite{Nenciu1993,Jung2000,ElgartHagedorn2012}.  Compared to these Gevrey-class clock approaches that achieve a $\delta$-dependent overhead of $\widetilde O\left(\log^q(\delta^{-1})\right)$ for $q>1 $, the Denjoy--Carleman clock we use achieves the near-optimal scaling $\widetilde O\left(\log(\delta^{-1})\right)$.

\section{Adiabatic Quantum Phase Estimation}
Let $H_S$ be the Hamiltonian of interest acting on an $n$-qubit system Hilbert space $\mathcal{H}_S$ and assume, after shifting and rescaling, that $\operatorname{spec}(H_S)\subset[0,1)$. The runtimes throughout are measured in physical time defined by this normalization. Fix a target additive accuracy $\eps>0$ and set
\begin{equation}
 m=\left\lceil\log_2\frac1\eps\right\rceil,
 \qquad
 \nu=2^{-m-1}.
\label{eq:m_nu_def}
\end{equation}
Let $\mathcal A$ be an $m$-qubit \textbf{a}ncilla register with qubits $a_1,\dots,a_m$, initialized in $\ket{+}^{\otimes m}$ with $\ket{+}=(\ket{0}+\ket{1})/\sqrt{2}$.
Define the counting operators
\begin{equation}
B_{j-1}\equiv \sum_{\ell=1}^{j-1}2^{-\ell} n_\ell,
\qquad n_\ell = \frac{\id - Z_{a_\ell}}{2},
\label{eq:prefix}
\end{equation} and the \emph{midpoint operator} at stage $j$ as
\begin{equation}
\mu_j \equiv B_{j-1}+2^{-j}
\label{eq:midpoint}
\end{equation}
acting on the ancilla register. Any computational basis string $b\in\{0,1\}^m$ is an eigenstate of $B_m$ with eigenvalue
$\val(b)=\sum_{\ell=1}^m 2^{-\ell}b_\ell$.
We use the center-valued estimate
\begin{equation}
\widehat\lambda(b)=\val(b)+2^{-m-1}
\label{eq:center_estimate}
\end{equation}
to associate a number to a computational basis state.
%The theorem below proves a measurement-success guarantee: after measuring the register, the estimate is $\eps$-accurate with probability at least $1-\delta$. The coherent map \eqref{eq:QPEact} can be obtained, up to constant-factor overhead, by copying the estimate to a clean output register and reversing the adiabatic evolution to cancel eigenvalue-dependent dynamical phases.

%The theorem below first proves a measurement-success guarantee: after measuring the register, the estimate is $\eps$-accurate with probability at least $1-\delta$. A genuine single-string coherent map of the form \eqref{eq:QPEact} requires the usual guard-band promise that no eigenvalue lies too close to an $m$-bit bin boundary; under this promise, the copy-and-reverse construction after Theorem~\ref{thm:runtime_correctness} implements the coherent map with the same asymptotic interrogation time.

\paragraph{Algorithm.}
Our adiabatic algorithm consists of $m$ sequential stages, each resolving one bit of the eigenvalues of $H_S$.  
At stage $j$ we evolve under the Hamiltonian
\begin{align}
H^{(j)}(s)
&=
-(1-s)\,X_{a_j}+s\left(H_S-\mu_j\right)Z_{a_j}.
\label{eq:compH}
\end{align}
The protocol is:
\begin{enumerate}
\item Initialize the system in an arbitrary state $\ket{\psi}_S=\sum_i c_i\ket{i}$ and the ancilla in $\ket{+}^{\otimes m}_{\mathcal A}$.
\item For $j=1,2,\dots,m$: evolve $s:0\to 1$ under $H^{(j)}(s)$ using the schedule $s(t)$ specified in Eq.~\eqref{eq:single_velocity}.
\item (Optional) Measure $\mathcal A$ in the computational basis to obtain a bit string $b$ and output the center estimate $\widehat\lambda(b)$ from Eq.~\eqref{eq:center_estimate}.
\end{enumerate}
Before proving correctness and efficiency of this protocol, we highlight several convenient properties. 
First, note that for all $j$ and $s$, $[H^{(j)}(s),H_S\otimes\id_\mathcal{A}]=0$, so the adiabatic evolution is block-diagonal in the eigenbasis of $H_S$. In particular, degenerate eigensectors of $H_S$ are preserved. Moreover, for $ \ell < j$, the ancilla qubit $a_\ell$ appears in $H^{(j)}(s)$ only through $Z_{a_\ell}$ terms, hence if it is set to a bit $\in \{0,1\}$ it remains conserved for all $j > \ell$. 
Thus the entire adiabatic evolution is block-diagonal across superpositions of energy eigenstates and across already-decided prefixes on the ancilla qubits $a_1,\dots, a_m$, and is thus described by an effective two-level Hamiltonian. 

\paragraph{Two-level reduction.}
Fix a system eigenstate $\ket{\alpha}$ with eigenvalue $\lambda\equiv\lambda_\alpha$, and fix a computational-basis prefix $x\in\{0,1\}^{j-1}$. On the joint subspace $\ket{\alpha}_S\otimes\ket{x}_{a_{<j}}\otimes\mathbb C^2_{a_j}$, the midpoint $\mu_j$ becomes a scalar $\mu_j(x)$, and $H^{(j)}(s)$ acts nontrivially only on $a_j$ as
\begin{equation}
h_\Delta(s)=-(1-s)X+s\,\Delta\,Z,
\qquad \Delta \equiv \lambda-\mu_j(x).
\label{eq:twolevel}
\end{equation}
At $s=0$ the unique ground state is $\ket{+}$. At $s=1$, the ground state is $\ket{1}$ if $\Delta>0$ and $\ket{0}$ if $\Delta<0$. Thus adiabatic evolution of the ancilla from $s=0$ to $s=1$ implements a threshold test: the sign of $\Delta$ determines whether the ground-state branch ends in $\ket{1}$ or $\ket{0}$. Repeating this threshold test for all $m$ stages coherently entangles the eigenvector $\ket{\alpha}$ with an $\eps$-approximation to the eigenvalue $\lambda_\alpha$. The above two-level reduction greatly simplifies the analysis of our algorithm and allows us to devise an explicit schedule that near-exponentially suppresses diabatic leakage.

%total system-ancilla Hamiltonian [refs], which is primary source of noise in a system weakly coupled to environment [refs]. Additionally, our protocol stores eigenvalues in the population of quantum states and not in fragile complex phases, making it robust to certain natural errors that drive us out of the groundstate: Z-basis dephasing of already-written output bits is much less harmful than in Fourier-based QPE. Running the schedule slightly longer does not destroy the information. 

%This adiabatic QPE is 
%This approach is useful if $H_S$ is available as a physical analog Hamiltonian in the lab and one can implement $Z \otimes H_S$. This allows us to skip the controlled digital simulation and compilation. Compare to von-Neumann: more ancillas, and if adiabatic you have to square. Wheeler-de-Witt (no!)

%Robust against dephasing in energy eigenbasis of the whole system. 

%Non-trivial approach but set up Lagrange equations for optimal schedule. 

%It avoids the digital QFT. 
For small $\Delta$, there seems to be an apparent gap closing in Eq. \eqref{eq:twolevel}. This is avoided by the fact that our algorithm only has to resolve eigenvalues up to accuracy $\eps \approx 2\nu$. 
More specifically, our analysis distinguishes between
(i) \emph{resolved branches} $|\Delta|>\nu$, where the sign of $\Delta$ must be determined reliably, and
(ii) \emph{unresolved branches} $|\Delta|\le\nu$, where the eigenvalue is effectively indistinguishable from the midpoint at this resolution.  We only need to bound diabatic leakage in the resolved case: in the unresolved regime, both bit-choices remain valid $\eps$-estimates of the true eigenvalue. The following two sections describe our schedule $s(t)$ and make the above error analysis in the resolved and unresolved case concrete. 
\section{Adiabatic Schedule}
To guarantee adiabatic evolution, we require a schedule $s(t)$ that evolves slowly relative to the spectral gap of the two-level Hamiltonian \eqref{eq:twolevel}. This gap is
\begin{equation}
E_e(s)-E_g(s)=2r_\Delta(s),
\label{eq:half_gap}
\end{equation}where $r_x(s)=\sqrt{(1-s)^2+s^2x^2}$ for a general placeholder variable $x$. To get around the suboptimal scaling in $\eps$ implied by \eqref{eq:adiabatic_simple_intro}, we employ a locally varying adiabatic schedule. The simplest candidate is the Roland-Cerf schedule, originally described in~\cite{RolandCerf2002} to recover Grover's quadratic speedup adiabatically. 

\paragraph{Roland-Cerf schedule.} 
The Roland-Cerf schedule is given by (visualized in Fig.~\ref{fig:schedule_trio})
\begin{equation}
\dot s(t) = 4\eta\,r_{\nu}(s(t))^2 =4 \eta \left((1-s)^2+\nu^2 s^2\right),
\label{eq:schedule}
\end{equation}
chosen such that the velocity approximately matches the square of the spectral gap of $h_\Delta$. %Here $\nu = 2^{-m-1}$ is a lower bound on the precision to which eigenvalues need to be resolved. 
The free parameter $\eta$ is chosen later to guarantee adiabaticity and will only depend on the overall failure probability $\delta$.
The duration of a single stage is  
\begin{equation}
T_\mathrm{stage}= \int_0^1 \frac{d s}{4\eta r_{\nu}(s)^2}=\frac{\pi}{8\eta \nu} = O(\eps^{-1}).
\label{eq:Tj}
\end{equation}
This yields an overall runtime  of $T_\mathrm{tot} = m T_\mathrm{stage}= \Theta( \eps^{-1} \log \eps^{-1})$, which is already Heisenberg-optimal in $\eps$ up to a logarithmic factor. 
However, the Roland-Cerf schedule suppresses diabatic leakage only polynomially (cf. Figure \ref{fig:leakage}): it is easy to verify that the total runtime scales with the failure probability $\delta$ as $T_\mathrm{tot} = \Theta(\delta^{-1/2})$, exponentially worse than the optimal bound of $T_\mathrm{tot} = \Theta(\log \delta^{-1})$~\cite{mande2023tight} achieved by gate-based QPE. This is because the speed \eqref{eq:schedule} and its derivatives are large at the endpoints of the schedule (cf. Fig.~\ref{fig:schedule_trio}), causing amplitude leakages into the excited state of order $O(1/T)$. 

\paragraph{Superadiabatic Roland-Cerf schedule.}

To obtain exponentially stronger bounds on the diabatic leakage, we instead reparametrize the Roland-Cerf schedule~\cite{RolandCerf2002} with a non-quasianalytic Denjoy--Carleman function class clock~\cite{Rainer2022,KrieglMichorRainer2009}. Define the Roland--Cerf action coordinate
\begin{equation}
\Theta_\nu(s)
\equiv
\int_0^s \frac{dz}{r_\nu(z)^2},
\qquad
r_\nu(s)=\sqrt{(1-s)^2+\nu^2s^2}.
\label{eq:Theta_def}
\end{equation}
A direct evaluation gives
\begin{equation}
\Theta_\nu(s)
=
\frac1\nu
\left[
\arctan\!\left(\frac{(1+\nu^2)s-1}{\nu}\right)
+
\arctan\!\left(\frac1\nu\right)
\right].
\label{eq:Theta_closed}
\end{equation} %where 
%\begin{equation}
%    \Theta_\nu(1)=\frac{\pi}{2\nu}.
%\end{equation}
\begin{figure}[t]
    \centering
    \includegraphics[width=\columnwidth]{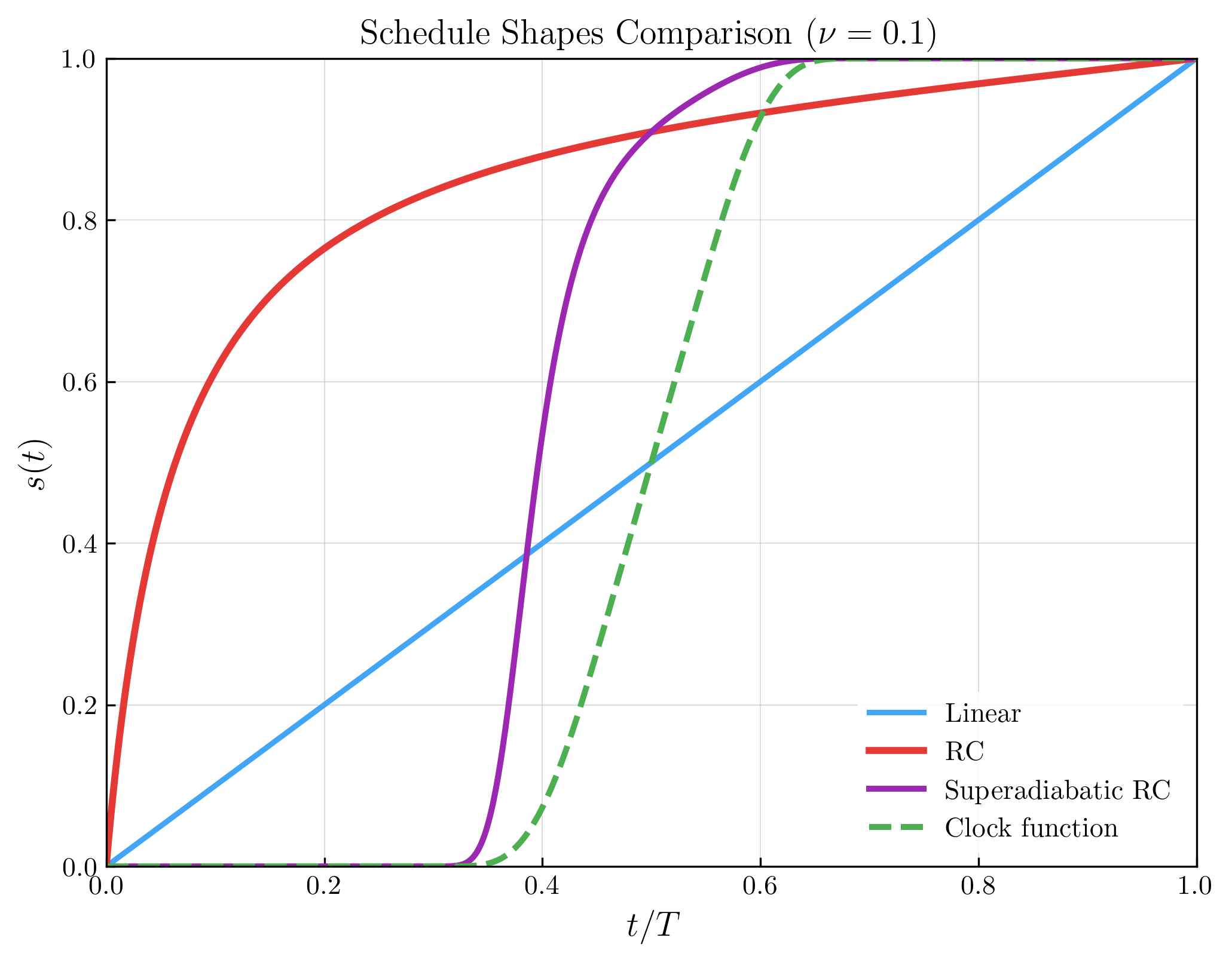}
    \caption{A visualization of three different schedule profiles:
    (a) The linear ramp profile $s(t) = t/T$. 
    (b) The Roland-Cerf schedule \eqref{eq:schedule} mirroring the behavior of the spectral gap.
    (c)  The Roland-Cerf schedule reparametrized by a DC clock, which we call superadiabatic RC. Also shown is the behavior of the DC clock function, whose endpoint derivatives vanish and which thereby near-exponentially suppresses diabatic transitions. 
    }
    \label{fig:schedule_trio}
\end{figure}
The ordinary Roland--Cerf schedule is constant-speed motion in the coordinate $\Theta_\nu(s)$.  Let $q>1$ be fixed and let $F_q:[0,1]\to[0,1]$ be monotone, with $F_q(u)=0$ near $0$, $F_q(u)=1$ near $1$, and
\begin{equation}
 \|F_q^{(n)}\|_\infty
 \le C_F R_F^n n![\log(n+e)]^{qn}
\qquad (n\ge 1).
\label{eq:DC_switch_main}
\end{equation}
The existence of such Denjoy--Carleman clock functions is proved constructively in Lemma~\ref{lem:dc_switch}, and visualized in Fig. \ref{fig:schedule_trio}. For a dimensionless adiabatic parameter $A = \eta^{-1}>0$, define the duration of one stage by
\begin{equation}
T_{\rm stage}(A)
=
\frac{A}{4}\Theta_\nu(1)
=
\frac{\pi A}{8\nu},
\label{eq:single_stage_time}
\end{equation}
and define the schedule implicitly by
\begin{equation}
\Theta_\nu(s(t))
=
\Theta_\nu(1)F_q\!\left(\frac{t}{T_{\rm stage}(A)}\right)
\label{eq:single_schedule}
\end{equation}
for $
0\le t\le T_{\rm stage}(A).$
Differentiating \eqref{eq:single_schedule} gives
\begin{equation}
\dot s(t)
=
\frac{4}{A}\,r_\nu(s(t))^2\,
F_q'\!\left(\frac{t}{T_{\rm stage}(A)}\right).
\label{eq:single_velocity}
\end{equation}
Thus the schedule retains the Roland--Cerf local-gap factor $r_\nu(s)^2$, but the clock is constant near both endpoints.  Equivalently, all positive time derivatives of the Hamiltonian vanish at the beginning and end of each stage, which removes the boundary terms that would otherwise appear in the high-order adiabatic expansion. These boundary cancellations allow us to prove near-exponentially small bounds on the diabatic leakage, significantly stronger than generically available bounds~\cite{lidar2009adiabatic}, as we show now. 

\paragraph{Diabatic leakage for superadiabatic RC.} 
Consider the two-level Hamiltonian \eqref{eq:twolevel} for a resolved branch $|\Delta|>\nu$. %As discussed earlier and made precise in the error analysis below, diabatic leakage does not need to be bounded in the unresolved case.  
Let $P_\Delta(s)$ denote the instantaneous ground-state projector of $h_\Delta(s)$.  The following theorem bounds transitions from the initial ground-state sector to the final excited sector.
\begin{theorem}[Near-exponentially small diabatic leakage]
\label{thm:main_leakage}
Fix $q>1$ and let $\mathcal{V}_{\Delta,A}^{\rm stage}$ be the single-stage two-level adiabatic evolution generated by the Hamiltonian \eqref{eq:twolevel} together with the superadiabatic RC schedule \eqref{eq:single_schedule}.  There exist constants $C,c>0$, depending only on $q$, such that for every $\nu\in(0,1/2]$, every resolved detuning $1\ge |\Delta|>\nu$, and every $A\ge 1$,
\begin{equation}
\bigl\|
(I-P_\Delta(1))\mathcal{V}_{\Delta,A}^{\rm stage}P_\Delta(0)
\bigr\|
\le
C\exp\!\left[-c\frac{A}{[\log(e+A)]^q}\right].
\label{eq:main_leakage_bound}
\end{equation}
\end{theorem}
The proof is given in Appendix~\ref{app:proof-thm1}.  The main technical input is the weighted superadiabatic estimate, Theorem~\ref{thm:weighted_global}, applied to the normalized two-level Hamiltonian in the Roland--Cerf action coordinate.
\begin{figure}[h!]
    \centering
    \includegraphics[width=\columnwidth]{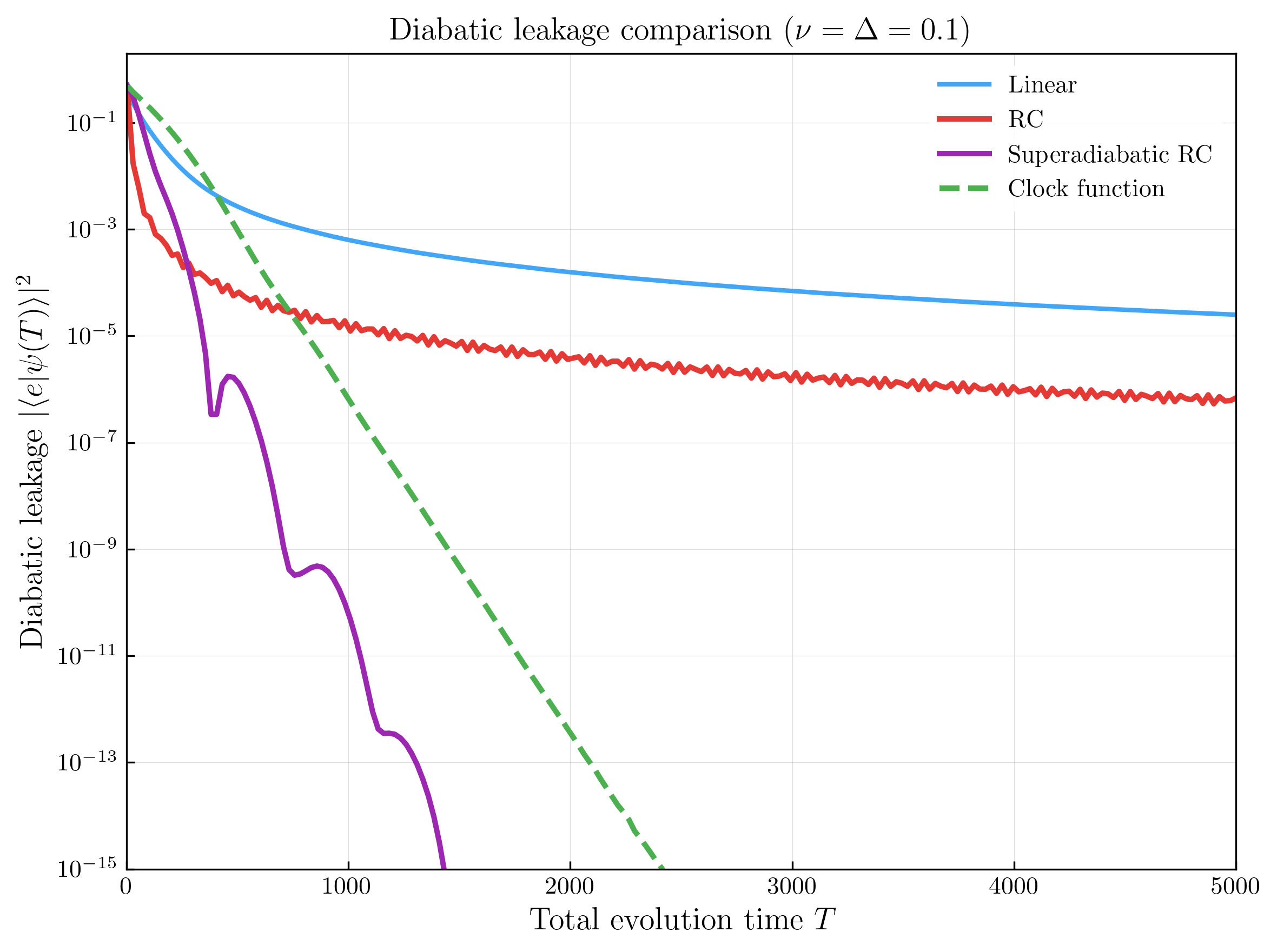}
    \caption{Numerical computation of the diabatic leakage via exact diagonalization of the two-level Hamiltonian \eqref{eq:twolevel}. Four different schedules are considered: linear, Roland-Cerf,  superadiabatic Roland-Cerf, and the DC clock function, which can be interpreted as the linear schedule reparametrized by the DC clock. While the linear and Roland-Cerf schedules yield a power-law decay, the reparametrized schedules suppress diabatic transitions near-exponentially. Both RC-type schedules display Stückelberg oscillations arising from coherent interference between nonadiabatic transitions.
    }
    \label{fig:leakage}
\end{figure}

We note that the same proof gives a high-confidence version of local adiabatic Grover search \cite{RolandCerf2002}, as we show in Appendix \ref{app:grover}.
Indeed, after symmetry reduction the Grover Hamiltonian is an affine two-level
avoided crossing with minimum gap $\sqrt{M/N}$, and the same superadiabatic Roland-Cerf schedule \eqref{eq:single_schedule} applies.  %Replacing the linear RC clock by the
%same endpoint-flat Denjoy--Carleman clock therefore also gives single-run failure
%probability
%\[
%    p_{\rm fail}
%    \le
%    C\exp[-cA/\log^q(e+A)]
%\]
%in time $T=A\,\Theta(\sqrt{N/M})$, for known $M$.  This improves the %confidence
%dependence of one coherent adiabatic Grover run, although ordinary measured
%Grover can also obtain logarithmic confidence overhead by repetition and
%verification.
\section{Error Analysis}
We now show that the protocol specified by the Hamiltonians \eqref{eq:twolevel} and the superadiabatic RC schedule \eqref{eq:single_schedule} evolves adiabatically and accurately implements QPE.  %The na\"ive local adiabatic criterion \eqref{eq:adiabatic_simple_intro} is not sufficient for the desired high-confidence scaling, so we exploit the two-level reduction and prove a superadiabatic leakage estimate directly.
To formalize this claim, let \[
H_S=\sum_\alpha \lambda_\alpha \Pi_\alpha
\]
be the spectral decomposition of \(H_S\), with \(\lambda_\alpha \in [0,1)\). Because all stage Hamiltonians \eqref{eq:compH} commute with $H_S$, degenerate eigenspaces of $H_S$ are preserved.
For each eigenvalue \(\lambda_\alpha\), let
\[
G^{\lambda_\alpha}
=
\sum_{\substack{b\in\{0,1\}^m\\
|\widehat\lambda(b)-\lambda_\alpha|\le\eps}}
|b\rangle\langle b|_\mathcal{A} 
\] be the projector onto valid $m$-bit estimates of $\lambda_\alpha$. Then the probability of failure for our algorithm is specified by the block diagonal spectral projector \begin{equation}
    \Pi_\mathrm{fail} := 
I_{\mathcal{S} \otimes \mathcal{A}}-
\sum_\alpha \Pi_\alpha\otimes G^{\lambda_\alpha}
\label{eq:fail}
\end{equation}
applied to the output state of our protocol \[
\mathcal{V}_A\left(|\psi\rangle_\mathcal{S}|+\rangle_\mathcal{A}^{\otimes m}\right).
\] Here, $\mathcal{V}_A$ is the full adiabatic QPE algorithm, that is, sequential adiabatic evolution under Hamiltonians \eqref{eq:compH} and the single schedule \eqref{eq:single_schedule} (with parameter $A$) for all stages $j=1,2,\dots,m$.

\begin{theorem}[Correctness and efficiency]
\label{thm:correctness}
Fix \(q>1\), target accuracy
\(\eps\in(0,1/2]\), and error probability \(\delta\in(0,1)\). Let $\mathcal{V}_A$ be the unitary adiabatic evolution under Hamiltonians \eqref{eq:compH} and the schedule \eqref{eq:single_schedule} with parameter $A$. Then there exists a choice of $A$ given by
% (specified in Eq. \eqref{eq:A_inversion_app}) such that: 
\begin{equation}
A=C_qL[\log(e+L)]^q,
\qquad
L=1+\log\frac{C_0m}{\sqrt\delta},
\label{eq:A_choice_main}
\end{equation}
with constants $C_0,C_q$ sufficiently large depending only on $q$ such that:
\begin{itemize}
    \item For every $n$-qubit input state \(|\psi\rangle_\mathcal{S}\),
\[
\left\|
\Pi_\mathrm{fail}
\mathcal V_A
\left(
|\psi\rangle_\mathcal{S}|+\rangle_\mathcal{A}^{\otimes m}
\right)
\right\|
\le
\sqrt\delta .
\]
    \item The total interrogation time of $\mathcal{V}_A$ is
\begin{align*}
T_{\rm tot}
&=
O\!\left(
\frac{\log(1/\eps)}{\eps}
\left[
1+\log\frac1\delta+\log\log\frac e\eps
\right]\right. \\[-0.25em]
&\hspace{1.2cm}\left.\times
\left[
\log\!\left(
e+\log\frac1\delta+\log\log\frac e\eps
\right)
\right]^q
\right) \\ &
=
\widetilde O\!\left(\frac1\eps\log\frac1\delta\right).
\label{eq:runtime_main}
\end{align*}
\end{itemize}
\end{theorem}
%Degeneracies of the Hamiltonian are naturally handled by the fact that $[H^{(j)}(s),H_S\otimes\id_\mathcal{A}]=0$.
The proof is given in Appendix~\ref{app:proof-thm2}.  It combines Theorem~\ref{thm:main_leakage} with a good-prefix invariant over the \(m=\lceil\log_2(1/\eps)\rceil\) midpoint comparisons and then sums the stage times.

\section{Discussion}
This section compares adiabatic QPE to alternative implementations and discusses its practicality, robustness, and coherence. In general, adiabatic QPE is naturally applicable if the Hamiltonian $H_S$ is available as a physical Hamiltonian in the lab via the interaction $ H_S \otimes Z$, in which case adiabatic QPE avoids expensive compiling of controlled powers of \(e^{-iH_St}\) through product formulas, block encodings, or
qubitization~\cite{Lloyd1996,BerryChildsKothari2015,LowChuang2019}. 

\paragraph{Robustness.}

Adiabatic QPE enjoys useful robustness properties.
First, it is well-known that adiabatic quantum computation is naturally resistant to dephasing errors in the energy eigenbasis~\cite{ChildsFarhiPreskill2002,
SarandyLidar2005,AlbashBoixoLidarZanardi2012,AlbashLidar2015,AlbashLidar2018}, which is an important source of noise in systems weakly coupled to the environment~\cite{paz1999quantum}. Second, unlike gate-based QPE which writes eigenvalues into relative
Fourier phases before the final QFT, our protocol stores  eigenvalue estimates
in quantum state populations
of the ancilla register. Since later stages interact with the ancilla register only through \(Z\)-diagonal terms, \(Z\)-basis dephasing of already written eigenvalue bits does not
affect the final measurement statistics and commutes with all subsequent
stages. Naturally, this observation does not guarantee robustness to other errors, such as relaxation, bit flips, dephasing of the active qubit in the wrong basis,
and bath-induced transitions between instantaneous eigenstates.

\paragraph{Hardware assumptions.} From a hardware perspective, the Hamiltonians \eqref{eq:compH} require only (i) a transverse field on the register, (ii) pairwise $Z_iZ_j$ interactions between register qubits, and (iii) a coupling $H_SZ_{a_j}$ between the system energy and a single ancilla qubit. This distinguishes AQPE from a non-adiabatic but analog
von-Neumann pointer measurement of energy, which would couple \(H_S\) to a
many-level pointer or clock degree of freedom through an interaction such as
\(H_S\otimes P\)~\cite{vonNeumann1955}.  Our protocol is closer in spirit to iterative or
semiclassical forms of phase estimation, which can also use a single control
qubit at a time~\cite{Kitaev1995,GriffithsNiu1996,lin2025analog}. The difference is that adiabatic QPE runs without intermediate measurements or classical feed-forward, and can be implemented coherently.

\paragraph{Previous adiabatic approaches.} To the best of our knowledge, no prior adiabatic implementation of QPE achieves the near-optimal scaling $T = \widetilde{O} \left(\frac{1}{\eps} \log\frac{1}{\delta}\right)$. 
Ref.~\cite{torosov2009design} constructs an adiabatic implementation of the
Quantum Fourier Transform (QFT) via circulant Hamiltonians and combines it with gate-based circuits to construct the QPE. This approach is hybrid and the precise runtime is not analyzed. Hen~\cite{hen2014fourier} constructs adiabatic implementations of several two-qubit gates and uses them to assemble an adiabatic
gate-by-gate realization of the QFT and thus QPE, but this approach inherits the circuit structure of gate-based QPE, and moreover requires relative phases to be maintained throughout the
adiabatic gates. 
%. Moreover, the construction involves degenerate ground-state manifolds and requires relative phases to be maintained throughout the adiabatic gates. 
Ref.~\cite{ZhangWeiPapageorgiou2010} gives an adiabatic
phase-estimation-like construction for quantum counting based on Berry phases, achieving $O(\eps^{-3/2})$ scaling, which is suboptimal relative to the $O(\eps^{-1})$
scaling of gate-based quantum counting and highlights the difficulty of recovering optimal scaling adiabatically. %Finally, Refs.~\cite{costa2022optimal,an2022quantum} do not study QPE, but show that
%carefully optimized adiabatic schedules can recover near-optimal or optimal algorithmic scalings
%for quantum linear-system solvers. 

\paragraph{Coherence.}
Our adiabatic evolution is block diagonal in the eigenspaces of \(H_S\), but as written it need not implement the ideal coherent map \eqref{eq:QPEact} with the same phases: different eigenvalue and prefix branches can accumulate different dynamical phases during the adiabatic evolution. (Geometric phases do not play a role because each two-level Hamiltonian \eqref{eq:twolevel} is real.) This does not affect applications of QPE if they rely only on the measurement statistics (such as Shor's algorithm~\cite{Shor1994}), or if they only use QPE to apply classical functions $f(\lambda)$ to the eigenvalues where $f$ does not vary strongly over $\eps-$valid estimates (such as $f(\lambda) = \lambda^{-1}$ for linear system solvers \cite{HHL2009}). In the latter case, simply apply the adiabatic protocol $\mathcal{V}$, apply a diagonal operation $f$ on the eigenvalue register, and subsequently reverse the adiabatic evolution which uncomputes the dynamical phases. 

If desired, a genuinely coherent implementation of QPE can be achieved by applying the adiabatic evolution $\mathcal{V}$, copying the eigenvalues from the ancilla register $\mathcal{A}$ to a separate register $\mathcal{R}$ with $\mathrm{CNOT}$ gates, and subsequently implementing the inverse adiabatic evolution $\mathcal{V}^\dagger$ on only the system $\mathcal{S}$ and the ancilla $\mathcal{A}$. This coherently uncomputes all dynamical phases while leaving eigenvectors entangled with corresponding eigenvalue estimates, \emph{provided} each eigenvector $\ket{\alpha}$ is associated with a single eigenvalue estimate $\widetilde\lambda_\alpha$. This is the case if no eigenvalue lies too close to a dyadic $m$-bit boundary, such that all eigenvalues are resolved. Alternatively, one can overresolve to \(m+r\) bits and copy only an
\(m\)-bit coarse label, thereby
reducing the probability that the copied label distinguishes two
valid boundary branches. 
\paragraph{Imperfect implementations of the schedule.}
The superadiabatic schedule in \eqref{eq:single_schedule} is a mathematically convenient schedule that 
should be viewed as a compact asymptotic way to package a family of finite-order endpoint-flat controls. For any fixed target \(\delta\), the leakage bound proof only uses integration-by-parts/superadiabatic orders up to \(N=O(\log(1/\delta)+\log\log(1/\eps))\). Thus an implementation need not realize an infinitely differentiable clock with infinite precision; it suffices to approximate the switch and its first \(N\) derivatives to the accuracy required by the target error budget. %If the waveform generator is subject to white noise, this does not distort the initial exponential decay of diabatic leakage visible in Fig \ref{fig:leakage}, but the leakage will necessarily hit a noise floor whose height is governed by the noise strength. 

In practice, one might want to replace the superadiabatic schedule with a schedule that is easier to implement with noisy waveform generators available in the lab. Many different simpler schedule choices exist that achieve nearly the same scaling:
One approach is to re-parametrize the Roland-Cerf schedule by a smooth Gevrey-class clock~\cite{Nenciu1993}, which achieves total interrogation time $T_{\rm tot}=
\widetilde O\!\left(\frac1\eps\log^q\frac1\delta\right)$ for any $q > 1$. Indeed, we have found that the simple Gevrey-class clock at $q=2$, while asymptotically worse, performs better than the DC-clock for realistic choices of $\delta$. If the initial state $\ket{\psi}_S$ is already close to an eigenstate, one could also directly use the simple Roland-Cerf schedule and exponentially amplify the success probability by repetition and majority vote. 

\section{Conclusion}
We have presented an adiabatic analogue of Quantum Phase Estimation: instead of implementing controlled powers of $e^{-iH_S t}$ and a quantum Fourier transform, we implement $m$ adiabatic midpoint comparisons that write an additive estimate directly into a register.
Our protocol requires only (i) a transverse field on the register, (ii) pairwise $Z_iZ_j$ interactions between register qubits, and (iii) a coupling $H_SZ_{j}$ between the system energy and a single ancilla qubit at a time. Our algorithm is further robust against certain dephasing errors. %as well as against imperfect control of the adiabatic clock schedule.

By reducing the analysis of adiabatic QPE to a two-level system, in which we reparametrize the Roland-Cerf schedule with a Denjoy-Carleman clock, we prove a total time complexity of
$T = \widetilde O(\eps^{-1}\log(1/\delta))$, which is optimal (up to logarithmic factors) in all parameters. 

This scheduling idea is not specific to quantum phase estimation.  We expect it to apply to general
adiabatic algorithms whose relevant dynamics reduce, exactly or perturbatively,
to a family of uniformly controlled avoided crossings.  The clearest example is
local adiabatic Grover search: the usual Roland--Cerf schedule already gives the
optimal $\Theta(\sqrt{N/M})$ search time, while reparameterizing the linear
Roland--Cerf clock by a Denjoy--Carleman clock suppresses
the single-run diabatic failure probability exponentially in the
dimensionless adiabatic parameter. Similar remarks apply to fixed-point
adiabatic search and adiabatic amplitude-amplification/state-conversion
problems. Our result thus suggests a general route to recovering the complexity of optimal gate-based quantum algorithms adiabatically.

%combine locally optimal gap-dependent schedules with endpoint-flat superadiabatic clocks to suppress diabatic error without giving up the optimal runtime scaling.

%Overall, w QPE can be implemented near-optimally adiabatically.

%We do not expect such a statement to hold for arbitrary
%many-level adiabatic algorithms without additional control of the normalized
%eigenpath geometry.

%We expect our schedule to find applications in adiabatic algorithms beyond the QPE-setting discussed here.

%The key technical result is a
%uniform near-exponential leakage bound  under
%a local action-coordinate Denjoy--Carleman schedule, yielding  %Our result shows that Heisenberg-limited QPE can be achieved adiabatically when the Hamiltonian coupling $Z\otimes H_S$ is native.

\begin{acknowledgments}
We thank Eddie Farhi, Alioscia Hamma, Aram Harrow, and Rolando Somma for helpful discussions. A.S. is supported by a Google PhD Fellowship and the Simons Foundation (MP-SIP-00001553, AWH).
S.L. is supported by, or in part by, the U. S. Army Research Laboratory and the U. S. Army Research Office under contract/grant number W911NF2310255, and by DoE under contract DE-SC0012704. 

\end{acknowledgments}

\bibliography{references}

%apsrev4-2.bst 2019-01-14 (MD) hand-edited version of apsrev4-1.bst
%Control: key (0)
%Control: author (8) initials jnrlst
%Control: editor formatted (1) identically to author
%Control: production of article title (0) allowed
%Control: page (0) single
%Control: year (1) truncated
%Control: production of eprint (0) enabled
\begin{thebibliography}{39}%
\makeatletter
\providecommand \@ifxundefined [1]{%
 \@ifx{#1\undefined}
}%
\providecommand \@ifnum [1]{%
 \ifnum #1\expandafter \@firstoftwo
 \else \expandafter \@secondoftwo
 \fi
}%
\providecommand \@ifx [1]{%
 \ifx #1\expandafter \@firstoftwo
 \else \expandafter \@secondoftwo
 \fi
}%
\providecommand \natexlab [1]{#1}%
\providecommand \enquote  [1]{``#1''}%
\providecommand \bibnamefont  [1]{#1}%
\providecommand \bibfnamefont [1]{#1}%
\providecommand \citenamefont [1]{#1}%
\providecommand \href@noop [0]{\@secondoftwo}%
\providecommand \href [0]{\begingroup \@sanitize@url \@href}%
\providecommand \@href[1]{\@@startlink{#1}\@@href}%
\providecommand \@@href[1]{\endgroup#1\@@endlink}%
\providecommand \@sanitize@url [0]{\catcode `\\12\catcode `\$12\catcode `\&12\catcode `\#12\catcode `\^12\catcode `\_12\catcode `\%12\relax}%
\providecommand \@@startlink[1]{}%
\providecommand \@@endlink[0]{}%
\providecommand \url  [0]{\begingroup\@sanitize@url \@url }%
\providecommand \@url [1]{\endgroup\@href {#1}{\urlprefix }}%
\providecommand \urlprefix  [0]{URL }%
\providecommand \Eprint [0]{\href }%
\providecommand \doibase [0]{https://doi.org/}%
\providecommand \selectlanguage [0]{\@gobble}%
\providecommand \bibinfo  [0]{\@secondoftwo}%
\providecommand \bibfield  [0]{\@secondoftwo}%
\providecommand \translation [1]{[#1]}%
\providecommand \BibitemOpen [0]{}%
\providecommand \bibitemStop [0]{}%
\providecommand \bibitemNoStop [0]{.\EOS\space}%
\providecommand \EOS [0]{\spacefactor3000\relax}%
\providecommand \BibitemShut  [1]{\csname bibitem#1\endcsname}%
\let\auto@bib@innerbib\@empty
%</preamble>
\bibitem [{\citenamefont {Kitaev}(1995)}]{Kitaev1995}%
  \BibitemOpen
  \bibfield  {author} {\bibinfo {author} {\bibfnamefont {A.~Y.}\ \bibnamefont {Kitaev}},\ }\href@noop {} {\bibinfo {title} {Quantum measurements and the {Abelian} stabilizer problem}} (\bibinfo {year} {1995}),\ \Eprint {https://arxiv.org/abs/quant-ph/9511026} {arXiv:quant-ph/9511026} \BibitemShut {NoStop}%
\bibitem [{\citenamefont {Cleve}\ \emph {et~al.}(1998)\citenamefont {Cleve}, \citenamefont {Ekert}, \citenamefont {Macchiavello},\ and\ \citenamefont {Mosca}}]{Cleve1998}%
  \BibitemOpen
  \bibfield  {author} {\bibinfo {author} {\bibfnamefont {R.}~\bibnamefont {Cleve}}, \bibinfo {author} {\bibfnamefont {A.}~\bibnamefont {Ekert}}, \bibinfo {author} {\bibfnamefont {C.}~\bibnamefont {Macchiavello}},\ and\ \bibinfo {author} {\bibfnamefont {M.}~\bibnamefont {Mosca}},\ }\bibfield  {title} {\bibinfo {title} {Quantum algorithms revisited},\ }\href@noop {} {\bibfield  {journal} {\bibinfo  {journal} {Proc. R. Soc. Lond. A}\ }\textbf {\bibinfo {volume} {454}},\ \bibinfo {pages} {339} (\bibinfo {year} {1998})}\BibitemShut {NoStop}%
\bibitem [{\citenamefont {Nielsen}\ and\ \citenamefont {Chuang}(2000)}]{NielsenChuang}%
  \BibitemOpen
  \bibfield  {author} {\bibinfo {author} {\bibfnamefont {M.~A.}\ \bibnamefont {Nielsen}}\ and\ \bibinfo {author} {\bibfnamefont {I.~L.}\ \bibnamefont {Chuang}},\ }\href@noop {} {\emph {\bibinfo {title} {Quantum Computation and Quantum Information}}}\ (\bibinfo  {publisher} {Cambridge University Press},\ \bibinfo {address} {Cambridge},\ \bibinfo {year} {2000})\BibitemShut {NoStop}%
\bibitem [{\citenamefont {Shor}(1994)}]{Shor1994}%
  \BibitemOpen
  \bibfield  {author} {\bibinfo {author} {\bibfnamefont {P.~W.}\ \bibnamefont {Shor}},\ }\bibfield  {title} {\bibinfo {title} {Algorithms for quantum computation: Discrete logarithms and factoring},\ }in\ \href@noop {} {\emph {\bibinfo {booktitle} {Proceedings of the 35th Annual Symposium on Foundations of Computer Science}}}\ (\bibinfo  {publisher} {IEEE},\ \bibinfo {year} {1994})\ pp.\ \bibinfo {pages} {124--134}\BibitemShut {NoStop}%
\bibitem [{\citenamefont {Abrams}\ and\ \citenamefont {Lloyd}(1999)}]{AbramsLloyd1999}%
  \BibitemOpen
  \bibfield  {author} {\bibinfo {author} {\bibfnamefont {D.~S.}\ \bibnamefont {Abrams}}\ and\ \bibinfo {author} {\bibfnamefont {S.}~\bibnamefont {Lloyd}},\ }\bibfield  {title} {\bibinfo {title} {Quantum algorithm providing exponential speed increase for finding eigenvalues and eigenvectors},\ }\href@noop {} {\bibfield  {journal} {\bibinfo  {journal} {Phys. Rev. Lett.}\ }\textbf {\bibinfo {volume} {83}},\ \bibinfo {pages} {5162} (\bibinfo {year} {1999})}\BibitemShut {NoStop}%
\bibitem [{\citenamefont {Aspuru-Guzik}\ \emph {et~al.}(2005)\citenamefont {Aspuru-Guzik}, \citenamefont {Dutoi}, \citenamefont {Love},\ and\ \citenamefont {Head-Gordon}}]{AspuruGuzik2005}%
  \BibitemOpen
  \bibfield  {author} {\bibinfo {author} {\bibfnamefont {A.}~\bibnamefont {Aspuru-Guzik}}, \bibinfo {author} {\bibfnamefont {A.~D.}\ \bibnamefont {Dutoi}}, \bibinfo {author} {\bibfnamefont {P.~J.}\ \bibnamefont {Love}},\ and\ \bibinfo {author} {\bibfnamefont {M.}~\bibnamefont {Head-Gordon}},\ }\bibfield  {title} {\bibinfo {title} {Simulated quantum computation of molecular energies},\ }\href@noop {} {\bibfield  {journal} {\bibinfo  {journal} {Science}\ }\textbf {\bibinfo {volume} {309}},\ \bibinfo {pages} {1704} (\bibinfo {year} {2005})}\BibitemShut {NoStop}%
\bibitem [{\citenamefont {Harrow}\ \emph {et~al.}(2009)\citenamefont {Harrow}, \citenamefont {Hassidim},\ and\ \citenamefont {Lloyd}}]{HHL2009}%
  \BibitemOpen
  \bibfield  {author} {\bibinfo {author} {\bibfnamefont {A.~W.}\ \bibnamefont {Harrow}}, \bibinfo {author} {\bibfnamefont {A.}~\bibnamefont {Hassidim}},\ and\ \bibinfo {author} {\bibfnamefont {S.}~\bibnamefont {Lloyd}},\ }\bibfield  {title} {\bibinfo {title} {Quantum algorithm for linear systems of equations},\ }\href@noop {} {\bibfield  {journal} {\bibinfo  {journal} {Phys. Rev. Lett.}\ }\textbf {\bibinfo {volume} {103}},\ \bibinfo {pages} {150502} (\bibinfo {year} {2009})}\BibitemShut {NoStop}%
\bibitem [{\citenamefont {Giovannetti}\ \emph {et~al.}(2006)\citenamefont {Giovannetti}, \citenamefont {Lloyd},\ and\ \citenamefont {Maccone}}]{GLM2006}%
  \BibitemOpen
  \bibfield  {author} {\bibinfo {author} {\bibfnamefont {V.}~\bibnamefont {Giovannetti}}, \bibinfo {author} {\bibfnamefont {S.}~\bibnamefont {Lloyd}},\ and\ \bibinfo {author} {\bibfnamefont {L.}~\bibnamefont {Maccone}},\ }\bibfield  {title} {\bibinfo {title} {Quantum metrology},\ }\href@noop {} {\bibfield  {journal} {\bibinfo  {journal} {Phys. Rev. Lett.}\ }\textbf {\bibinfo {volume} {96}},\ \bibinfo {pages} {010401} (\bibinfo {year} {2006})}\BibitemShut {NoStop}%
\bibitem [{\citenamefont {Lloyd}(1996)}]{Lloyd1996}%
  \BibitemOpen
  \bibfield  {author} {\bibinfo {author} {\bibfnamefont {S.}~\bibnamefont {Lloyd}},\ }\bibfield  {title} {\bibinfo {title} {Universal quantum simulators},\ }\href@noop {} {\bibfield  {journal} {\bibinfo  {journal} {Science}\ }\textbf {\bibinfo {volume} {273}},\ \bibinfo {pages} {1073} (\bibinfo {year} {1996})}\BibitemShut {NoStop}%
\bibitem [{\citenamefont {Berry}\ \emph {et~al.}(2015)\citenamefont {Berry}, \citenamefont {Childs},\ and\ \citenamefont {Kothari}}]{BerryChildsKothari2015}%
  \BibitemOpen
  \bibfield  {author} {\bibinfo {author} {\bibfnamefont {D.~W.}\ \bibnamefont {Berry}}, \bibinfo {author} {\bibfnamefont {A.~M.}\ \bibnamefont {Childs}},\ and\ \bibinfo {author} {\bibfnamefont {R.}~\bibnamefont {Kothari}},\ }\bibfield  {title} {\bibinfo {title} {Hamiltonian simulation with nearly optimal dependence on all parameters},\ }in\ \href@noop {} {\emph {\bibinfo {booktitle} {Proceedings of the 56th IEEE Symposium on Foundations of Computer Science}}}\ (\bibinfo  {publisher} {IEEE},\ \bibinfo {year} {2015})\ pp.\ \bibinfo {pages} {792--809}\BibitemShut {NoStop}%
\bibitem [{\citenamefont {Low}\ and\ \citenamefont {Chuang}(2019)}]{LowChuang2019}%
  \BibitemOpen
  \bibfield  {author} {\bibinfo {author} {\bibfnamefont {G.~H.}\ \bibnamefont {Low}}\ and\ \bibinfo {author} {\bibfnamefont {I.~L.}\ \bibnamefont {Chuang}},\ }\bibfield  {title} {\bibinfo {title} {Hamiltonian simulation by qubitization},\ }\href@noop {} {\bibfield  {journal} {\bibinfo  {journal} {Quantum}\ }\textbf {\bibinfo {volume} {3}},\ \bibinfo {pages} {163} (\bibinfo {year} {2019})}\BibitemShut {NoStop}%
\bibitem [{\citenamefont {Mande}\ and\ \citenamefont {de~Wolf}(2023)}]{mande2023tight}%
  \BibitemOpen
  \bibfield  {author} {\bibinfo {author} {\bibfnamefont {N.~S.}\ \bibnamefont {Mande}}\ and\ \bibinfo {author} {\bibfnamefont {R.}~\bibnamefont {de~Wolf}},\ }\bibfield  {title} {\bibinfo {title} {Tight bounds for quantum phase estimation and related problems},\ }in\ \href@noop {} {\emph {\bibinfo {booktitle} {31st Annual European Symposium on Algorithms (ESA 2023)}}},\ \bibinfo {series} {LIPIcs}, Vol.\ \bibinfo {volume} {274}\ (\bibinfo  {publisher} {Schloss Dagstuhl},\ \bibinfo {year} {2023})\ pp.\ \bibinfo {pages} {81:1--81:16}\BibitemShut {NoStop}%
\bibitem [{\citenamefont {Roland}\ and\ \citenamefont {Cerf}(2002)}]{RolandCerf2002}%
  \BibitemOpen
  \bibfield  {author} {\bibinfo {author} {\bibfnamefont {J.}~\bibnamefont {Roland}}\ and\ \bibinfo {author} {\bibfnamefont {N.~J.}\ \bibnamefont {Cerf}},\ }\bibfield  {title} {\bibinfo {title} {Quantum search by local adiabatic evolution},\ }\href@noop {} {\bibfield  {journal} {\bibinfo  {journal} {Phys. Rev. A}\ }\textbf {\bibinfo {volume} {65}},\ \bibinfo {pages} {042308} (\bibinfo {year} {2002})}\BibitemShut {NoStop}%
\bibitem [{\citenamefont {Aharonov}\ \emph {et~al.}(2007)\citenamefont {Aharonov}, \citenamefont {van Dam}, \citenamefont {Kempe}, \citenamefont {Landau}, \citenamefont {Lloyd},\ and\ \citenamefont {Regev}}]{Aharonov2007}%
  \BibitemOpen
  \bibfield  {author} {\bibinfo {author} {\bibfnamefont {D.}~\bibnamefont {Aharonov}}, \bibinfo {author} {\bibfnamefont {W.}~\bibnamefont {van Dam}}, \bibinfo {author} {\bibfnamefont {J.}~\bibnamefont {Kempe}}, \bibinfo {author} {\bibfnamefont {Z.}~\bibnamefont {Landau}}, \bibinfo {author} {\bibfnamefont {S.}~\bibnamefont {Lloyd}},\ and\ \bibinfo {author} {\bibfnamefont {O.}~\bibnamefont {Regev}},\ }\bibfield  {title} {\bibinfo {title} {Adiabatic quantum computation is equivalent to standard quantum computation},\ }\href@noop {} {\bibfield  {journal} {\bibinfo  {journal} {SIAM J. Comput.}\ }\textbf {\bibinfo {volume} {37}},\ \bibinfo {pages} {166} (\bibinfo {year} {2007})}\BibitemShut {NoStop}%
\bibitem [{\citenamefont {Gietka}\ \emph {et~al.}(2021)\citenamefont {Gietka}, \citenamefont {Metz}, \citenamefont {Keller},\ and\ \citenamefont {Li}}]{gietka2021adiabatic}%
  \BibitemOpen
  \bibfield  {author} {\bibinfo {author} {\bibfnamefont {K.}~\bibnamefont {Gietka}}, \bibinfo {author} {\bibfnamefont {F.}~\bibnamefont {Metz}}, \bibinfo {author} {\bibfnamefont {T.}~\bibnamefont {Keller}},\ and\ \bibinfo {author} {\bibfnamefont {J.}~\bibnamefont {Li}},\ }\bibfield  {title} {\bibinfo {title} {Adiabatic critical quantum metrology cannot reach the {Heisenberg} limit even when shortcuts to adiabaticity are applied},\ }\href@noop {} {\bibfield  {journal} {\bibinfo  {journal} {Quantum}\ }\textbf {\bibinfo {volume} {5}},\ \bibinfo {pages} {489} (\bibinfo {year} {2021})}\BibitemShut {NoStop}%
\bibitem [{\citenamefont {Lidar}\ \emph {et~al.}(2009)\citenamefont {Lidar}, \citenamefont {Rezakhani},\ and\ \citenamefont {Hamma}}]{lidar2009adiabatic}%
  \BibitemOpen
  \bibfield  {author} {\bibinfo {author} {\bibfnamefont {D.~A.}\ \bibnamefont {Lidar}}, \bibinfo {author} {\bibfnamefont {A.~T.}\ \bibnamefont {Rezakhani}},\ and\ \bibinfo {author} {\bibfnamefont {A.}~\bibnamefont {Hamma}},\ }\bibfield  {title} {\bibinfo {title} {Adiabatic approximation with exponential accuracy for many-body systems and quantum computation},\ }\href@noop {} {\bibfield  {journal} {\bibinfo  {journal} {J. Math. Phys.}\ }\textbf {\bibinfo {volume} {50}},\ \bibinfo {pages} {102106} (\bibinfo {year} {2009})}\BibitemShut {NoStop}%
\bibitem [{\citenamefont {Farhi}\ \emph {et~al.}(2000)\citenamefont {Farhi}, \citenamefont {Goldstone}, \citenamefont {Gutmann},\ and\ \citenamefont {Sipser}}]{farhi2000quantum}%
  \BibitemOpen
  \bibfield  {author} {\bibinfo {author} {\bibfnamefont {E.}~\bibnamefont {Farhi}}, \bibinfo {author} {\bibfnamefont {J.}~\bibnamefont {Goldstone}}, \bibinfo {author} {\bibfnamefont {S.}~\bibnamefont {Gutmann}},\ and\ \bibinfo {author} {\bibfnamefont {M.}~\bibnamefont {Sipser}},\ }\href@noop {} {\bibinfo {title} {Quantum computation by adiabatic evolution}} (\bibinfo {year} {2000}),\ \Eprint {https://arxiv.org/abs/quant-ph/0001106} {arXiv:quant-ph/0001106} \BibitemShut {NoStop}%
\bibitem [{\citenamefont {Farhi}\ \emph {et~al.}(2001)\citenamefont {Farhi}, \citenamefont {Goldstone}, \citenamefont {Gutmann}, \citenamefont {Lapan}, \citenamefont {Lundgren},\ and\ \citenamefont {Preda}}]{farhi2001quantum}%
  \BibitemOpen
  \bibfield  {author} {\bibinfo {author} {\bibfnamefont {E.}~\bibnamefont {Farhi}}, \bibinfo {author} {\bibfnamefont {J.}~\bibnamefont {Goldstone}}, \bibinfo {author} {\bibfnamefont {S.}~\bibnamefont {Gutmann}}, \bibinfo {author} {\bibfnamefont {J.}~\bibnamefont {Lapan}}, \bibinfo {author} {\bibfnamefont {A.}~\bibnamefont {Lundgren}},\ and\ \bibinfo {author} {\bibfnamefont {D.}~\bibnamefont {Preda}},\ }\bibfield  {title} {\bibinfo {title} {A quantum adiabatic evolution algorithm applied to random instances of an {NP}-complete problem},\ }\href@noop {} {\bibfield  {journal} {\bibinfo  {journal} {Science}\ }\textbf {\bibinfo {volume} {292}},\ \bibinfo {pages} {472} (\bibinfo {year} {2001})}\BibitemShut {NoStop}%
\bibitem [{\citenamefont {Du}\ \emph {et~al.}(2008)\citenamefont {Du}, \citenamefont {Hu}, \citenamefont {Wang}, \citenamefont {Wu}, \citenamefont {Zhao},\ and\ \citenamefont {Suter}}]{du2008experimental}%
  \BibitemOpen
  \bibfield  {author} {\bibinfo {author} {\bibfnamefont {J.}~\bibnamefont {Du}}, \bibinfo {author} {\bibfnamefont {L.}~\bibnamefont {Hu}}, \bibinfo {author} {\bibfnamefont {Y.}~\bibnamefont {Wang}}, \bibinfo {author} {\bibfnamefont {J.}~\bibnamefont {Wu}}, \bibinfo {author} {\bibfnamefont {M.}~\bibnamefont {Zhao}},\ and\ \bibinfo {author} {\bibfnamefont {D.}~\bibnamefont {Suter}},\ }\bibfield  {title} {\bibinfo {title} {Experimental study of the validity of quantitative conditions in the quantum adiabatic theorem},\ }\href@noop {} {\bibfield  {journal} {\bibinfo  {journal} {Phys. Rev. Lett.}\ }\textbf {\bibinfo {volume} {101}},\ \bibinfo {pages} {060403} (\bibinfo {year} {2008})}\BibitemShut {NoStop}%
\bibitem [{\citenamefont {Messiah}(1962)}]{messiah1962quantum}%
  \BibitemOpen
  \bibfield  {author} {\bibinfo {author} {\bibfnamefont {A.}~\bibnamefont {Messiah}},\ }\href@noop {} {\emph {\bibinfo {title} {Quantum Mechanics, Vol. II}}}\ (\bibinfo  {publisher} {North-Holland Publishing Company},\ \bibinfo {address} {Amsterdam},\ \bibinfo {year} {1962})\BibitemShut {NoStop}%
\bibitem [{\citenamefont {Rainer}(2022)}]{Rainer2022}%
  \BibitemOpen
  \bibfield  {author} {\bibinfo {author} {\bibfnamefont {A.}~\bibnamefont {Rainer}},\ }\bibfield  {title} {\bibinfo {title} {Ultradifferentiable extension theorems: A survey},\ }\href@noop {} {\bibfield  {journal} {\bibinfo  {journal} {Expo. Math.}\ }\textbf {\bibinfo {volume} {40}},\ \bibinfo {pages} {679} (\bibinfo {year} {2022})}\BibitemShut {NoStop}%
\bibitem [{\citenamefont {Kriegl}\ \emph {et~al.}(2009)\citenamefont {Kriegl}, \citenamefont {Michor},\ and\ \citenamefont {Rainer}}]{KrieglMichorRainer2009}%
  \BibitemOpen
  \bibfield  {author} {\bibinfo {author} {\bibfnamefont {A.}~\bibnamefont {Kriegl}}, \bibinfo {author} {\bibfnamefont {P.~W.}\ \bibnamefont {Michor}},\ and\ \bibinfo {author} {\bibfnamefont {A.}~\bibnamefont {Rainer}},\ }\bibfield  {title} {\bibinfo {title} {The convenient setting for non-quasianalytic superadiabatic differentiable mappings},\ }\href@noop {} {\bibfield  {journal} {\bibinfo  {journal} {J. Funct. Anal.}\ }\textbf {\bibinfo {volume} {256}},\ \bibinfo {pages} {3510} (\bibinfo {year} {2009})}\BibitemShut {NoStop}%
\bibitem [{\citenamefont {Nenciu}(1993)}]{Nenciu1993}%
  \BibitemOpen
  \bibfield  {author} {\bibinfo {author} {\bibfnamefont {G.}~\bibnamefont {Nenciu}},\ }\bibfield  {title} {\bibinfo {title} {Linear adiabatic theory. exponential estimates},\ }\href@noop {} {\bibfield  {journal} {\bibinfo  {journal} {Commun. Math. Phys.}\ }\textbf {\bibinfo {volume} {152}},\ \bibinfo {pages} {479} (\bibinfo {year} {1993})}\BibitemShut {NoStop}%
\bibitem [{\citenamefont {Jung}(2000)}]{Jung2000}%
  \BibitemOpen
  \bibfield  {author} {\bibinfo {author} {\bibfnamefont {K.}~\bibnamefont {Jung}},\ }\href {https://cds.cern.ch/record/441150} {\emph {\bibinfo {title} {The Adiabatic Theorem for Switching Processes with {Gevrey} Class Regularity}}},\ \bibinfo {type} {Tech. Rep.}\ \bibinfo {number} {Preprint 442}\ (\bibinfo  {institution} {Sonderforschungsbereich 288, Technische Universit{\"a}t Berlin},\ \bibinfo {year} {2000})\ \bibinfo {note} {cERN-EXT-2000-159}\BibitemShut {NoStop}%
\bibitem [{\citenamefont {Elgart}\ and\ \citenamefont {Hagedorn}(2012)}]{ElgartHagedorn2012}%
  \BibitemOpen
  \bibfield  {author} {\bibinfo {author} {\bibfnamefont {A.}~\bibnamefont {Elgart}}\ and\ \bibinfo {author} {\bibfnamefont {G.~A.}\ \bibnamefont {Hagedorn}},\ }\bibfield  {title} {\bibinfo {title} {A note on the switching adiabatic theorem},\ }\href {https://doi.org/10.1063/1.4748968} {\bibfield  {journal} {\bibinfo  {journal} {Journal of Mathematical Physics}\ }\textbf {\bibinfo {volume} {53}},\ \bibinfo {pages} {102202} (\bibinfo {year} {2012})},\ \Eprint {https://arxiv.org/abs/1204.2318} {arXiv:1204.2318 [math-ph]} \BibitemShut {NoStop}%
\bibitem [{\citenamefont {Childs}\ \emph {et~al.}(2002)\citenamefont {Childs}, \citenamefont {Farhi},\ and\ \citenamefont {Preskill}}]{ChildsFarhiPreskill2002}%
  \BibitemOpen
  \bibfield  {author} {\bibinfo {author} {\bibfnamefont {A.~M.}\ \bibnamefont {Childs}}, \bibinfo {author} {\bibfnamefont {E.}~\bibnamefont {Farhi}},\ and\ \bibinfo {author} {\bibfnamefont {J.}~\bibnamefont {Preskill}},\ }\bibfield  {title} {\bibinfo {title} {Robustness of adiabatic quantum computation},\ }\href {https://doi.org/10.1103/PhysRevA.65.012322} {\bibfield  {journal} {\bibinfo  {journal} {Physical Review A}\ }\textbf {\bibinfo {volume} {65}},\ \bibinfo {pages} {012322} (\bibinfo {year} {2002})}\BibitemShut {NoStop}%
\bibitem [{\citenamefont {Sarandy}\ and\ \citenamefont {Lidar}(2005)}]{SarandyLidar2005}%
  \BibitemOpen
  \bibfield  {author} {\bibinfo {author} {\bibfnamefont {M.~S.}\ \bibnamefont {Sarandy}}\ and\ \bibinfo {author} {\bibfnamefont {D.~A.}\ \bibnamefont {Lidar}},\ }\bibfield  {title} {\bibinfo {title} {Adiabatic quantum computation in open systems},\ }\href {https://doi.org/10.1103/PhysRevLett.95.250503} {\bibfield  {journal} {\bibinfo  {journal} {Physical Review Letters}\ }\textbf {\bibinfo {volume} {95}},\ \bibinfo {pages} {250503} (\bibinfo {year} {2005})}\BibitemShut {NoStop}%
\bibitem [{\citenamefont {Albash}\ \emph {et~al.}(2012)\citenamefont {Albash}, \citenamefont {Boixo}, \citenamefont {Lidar},\ and\ \citenamefont {Zanardi}}]{AlbashBoixoLidarZanardi2012}%
  \BibitemOpen
  \bibfield  {author} {\bibinfo {author} {\bibfnamefont {T.}~\bibnamefont {Albash}}, \bibinfo {author} {\bibfnamefont {S.}~\bibnamefont {Boixo}}, \bibinfo {author} {\bibfnamefont {D.~A.}\ \bibnamefont {Lidar}},\ and\ \bibinfo {author} {\bibfnamefont {P.}~\bibnamefont {Zanardi}},\ }\bibfield  {title} {\bibinfo {title} {Quantum adiabatic markovian master equations},\ }\href {https://doi.org/10.1088/1367-2630/14/12/123016} {\bibfield  {journal} {\bibinfo  {journal} {New Journal of Physics}\ }\textbf {\bibinfo {volume} {14}},\ \bibinfo {pages} {123016} (\bibinfo {year} {2012})}\BibitemShut {NoStop}%
\bibitem [{\citenamefont {Albash}\ and\ \citenamefont {Lidar}(2015)}]{AlbashLidar2015}%
  \BibitemOpen
  \bibfield  {author} {\bibinfo {author} {\bibfnamefont {T.}~\bibnamefont {Albash}}\ and\ \bibinfo {author} {\bibfnamefont {D.~A.}\ \bibnamefont {Lidar}},\ }\bibfield  {title} {\bibinfo {title} {Decoherence in adiabatic quantum computation},\ }\href {https://doi.org/10.1103/PhysRevA.91.062320} {\bibfield  {journal} {\bibinfo  {journal} {Physical Review A}\ }\textbf {\bibinfo {volume} {91}},\ \bibinfo {pages} {062320} (\bibinfo {year} {2015})}\BibitemShut {NoStop}%
\bibitem [{\citenamefont {Albash}\ and\ \citenamefont {Lidar}(2018)}]{AlbashLidar2018}%
  \BibitemOpen
  \bibfield  {author} {\bibinfo {author} {\bibfnamefont {T.}~\bibnamefont {Albash}}\ and\ \bibinfo {author} {\bibfnamefont {D.~A.}\ \bibnamefont {Lidar}},\ }\bibfield  {title} {\bibinfo {title} {Adiabatic quantum computation},\ }\href@noop {} {\bibfield  {journal} {\bibinfo  {journal} {Rev. Mod. Phys.}\ }\textbf {\bibinfo {volume} {90}},\ \bibinfo {pages} {015002} (\bibinfo {year} {2018})}\BibitemShut {NoStop}%
\bibitem [{\citenamefont {Paz}\ and\ \citenamefont {Zurek}(1999)}]{paz1999quantum}%
  \BibitemOpen
  \bibfield  {author} {\bibinfo {author} {\bibfnamefont {J.~P.}\ \bibnamefont {Paz}}\ and\ \bibinfo {author} {\bibfnamefont {W.~H.}\ \bibnamefont {Zurek}},\ }\bibfield  {title} {\bibinfo {title} {Quantum limit of decoherence: Environment induced superselection of energy eigenstates},\ }\href@noop {} {\bibfield  {journal} {\bibinfo  {journal} {Physical Review Letters}\ }\textbf {\bibinfo {volume} {82}},\ \bibinfo {pages} {5181} (\bibinfo {year} {1999})}\BibitemShut {NoStop}%
\bibitem [{\citenamefont {von Neumann}(1955)}]{vonNeumann1955}%
  \BibitemOpen
  \bibfield  {author} {\bibinfo {author} {\bibfnamefont {J.}~\bibnamefont {von Neumann}},\ }\href@noop {} {\emph {\bibinfo {title} {Mathematical Foundations of Quantum Mechanics}}}\ (\bibinfo  {publisher} {Princeton University Press},\ \bibinfo {year} {1955})\BibitemShut {NoStop}%
\bibitem [{\citenamefont {Griffiths}\ and\ \citenamefont {Niu}(1996)}]{GriffithsNiu1996}%
  \BibitemOpen
  \bibfield  {author} {\bibinfo {author} {\bibfnamefont {R.~B.}\ \bibnamefont {Griffiths}}\ and\ \bibinfo {author} {\bibfnamefont {C.-S.}\ \bibnamefont {Niu}},\ }\bibfield  {title} {\bibinfo {title} {Semiclassical fourier transform for quantum computation},\ }\href {https://doi.org/10.1103/PhysRevLett.76.3228} {\bibfield  {journal} {\bibinfo  {journal} {Physical Review Letters}\ }\textbf {\bibinfo {volume} {76}},\ \bibinfo {pages} {3228} (\bibinfo {year} {1996})}\BibitemShut {NoStop}%
\bibitem [{\citenamefont {Lin}\ and\ \citenamefont {Wang}(2025)}]{lin2025analog}%
  \BibitemOpen
  \bibfield  {author} {\bibinfo {author} {\bibfnamefont {W.-C.}\ \bibnamefont {Lin}}\ and\ \bibinfo {author} {\bibfnamefont {C.-H.}\ \bibnamefont {Wang}},\ }\bibfield  {title} {\bibinfo {title} {Analog quantum phase estimation with single-mode readout},\ }\href@noop {} {\bibfield  {journal} {\bibinfo  {journal} {arXiv preprint arXiv:2506.15668}\ } (\bibinfo {year} {2025})}\BibitemShut {NoStop}%
\bibitem [{\citenamefont {Torosov}\ and\ \citenamefont {Vitanov}(2009)}]{torosov2009design}%
  \BibitemOpen
  \bibfield  {author} {\bibinfo {author} {\bibfnamefont {B.~T.}\ \bibnamefont {Torosov}}\ and\ \bibinfo {author} {\bibfnamefont {N.~V.}\ \bibnamefont {Vitanov}},\ }\bibfield  {title} {\bibinfo {title} {Design of quantum fourier transforms and quantum algorithms by using circulant hamiltonians},\ }\href@noop {} {\bibfield  {journal} {\bibinfo  {journal} {Physical Review A—Atomic, Molecular, and Optical Physics}\ }\textbf {\bibinfo {volume} {80}},\ \bibinfo {pages} {022329} (\bibinfo {year} {2009})}\BibitemShut {NoStop}%
\bibitem [{\citenamefont {Hen}(2014)}]{hen2014fourier}%
  \BibitemOpen
  \bibfield  {author} {\bibinfo {author} {\bibfnamefont {I.}~\bibnamefont {Hen}},\ }\bibfield  {title} {\bibinfo {title} {Fourier-transforming with quantum annealers},\ }\href@noop {} {\bibfield  {journal} {\bibinfo  {journal} {Frontiers in Physics}\ }\textbf {\bibinfo {volume} {2}},\ \bibinfo {pages} {44} (\bibinfo {year} {2014})}\BibitemShut {NoStop}%
\bibitem [{\citenamefont {Zhang}\ \emph {et~al.}(2010)\citenamefont {Zhang}, \citenamefont {Wei},\ and\ \citenamefont {Papageorgiou}}]{ZhangWeiPapageorgiou2010}%
  \BibitemOpen
  \bibfield  {author} {\bibinfo {author} {\bibfnamefont {C.}~\bibnamefont {Zhang}}, \bibinfo {author} {\bibfnamefont {Z.}~\bibnamefont {Wei}},\ and\ \bibinfo {author} {\bibfnamefont {A.}~\bibnamefont {Papageorgiou}},\ }\bibfield  {title} {\bibinfo {title} {Adiabatic quantum counting by geometric phase estimation},\ }\href@noop {} {\bibfield  {journal} {\bibinfo  {journal} {Quantum Inf. Process.}\ }\textbf {\bibinfo {volume} {9}},\ \bibinfo {pages} {369} (\bibinfo {year} {2010})}\BibitemShut {NoStop}%
\bibitem [{\citenamefont {Hagedorn}\ and\ \citenamefont {Joye}(2002)}]{HagedornJoye2002}%
  \BibitemOpen
  \bibfield  {author} {\bibinfo {author} {\bibfnamefont {G.~A.}\ \bibnamefont {Hagedorn}}\ and\ \bibinfo {author} {\bibfnamefont {A.}~\bibnamefont {Joye}},\ }\bibfield  {title} {\bibinfo {title} {Elementary exponential error estimates for the adiabatic approximation},\ }\href@noop {} {\bibfield  {journal} {\bibinfo  {journal} {J. Math. Anal. Appl.}\ }\textbf {\bibinfo {volume} {267}},\ \bibinfo {pages} {235} (\bibinfo {year} {2002})}\BibitemShut {NoStop}%
\bibitem [{\citenamefont {Jansen}\ \emph {et~al.}(2007)\citenamefont {Jansen}, \citenamefont {Ruskai},\ and\ \citenamefont {Seiler}}]{JRS2007}%
  \BibitemOpen
  \bibfield  {author} {\bibinfo {author} {\bibfnamefont {S.}~\bibnamefont {Jansen}}, \bibinfo {author} {\bibfnamefont {M.-B.}\ \bibnamefont {Ruskai}},\ and\ \bibinfo {author} {\bibfnamefont {R.}~\bibnamefont {Seiler}},\ }\bibfield  {title} {\bibinfo {title} {Bounds for the adiabatic approximation with applications to quantum computation},\ }\href@noop {} {\bibfield  {journal} {\bibinfo  {journal} {J. Math. Phys.}\ }\textbf {\bibinfo {volume} {48}},\ \bibinfo {pages} {102111} (\bibinfo {year} {2007})}\BibitemShut {NoStop}%
\end{thebibliography}%

\newpage

\appendix
\onecolumngrid

\section{Proof overview and conventions}
\manualref{A.1}\label{app:dc_leakage}

This appendix proves the two estimates used in the main text. It is deliberately explicit to provide a detailed and self-contained proof of our main results. The argument has four ingredients.  First, we construct a compactly supported flat clock in a non-quasianalytic Denjoy--Carleman class.  Second, we prove a finite-dimensional superadiabatic estimate with endpoint cancellation.  This is a standard superadiabatic construction in the spirit of Refs.~\cite{Nenciu1993,HagedornJoye2002,JRS2007,lidar2009adiabatic}, with one modification: the residual is estimated in \(L^1\) so that the constants do not depend on the length of the rescaled gap-time interval.  Third, we apply these constructions to the normalized two-level comparator obtained from the Roland--Cerf action coordinate~\cite{RolandCerf2002}.  Finally, we compose the one-stage leakage estimate over the \(m\) midpoint-comparison stages.

Throughout the appendix, \(q>1\) is fixed and
\begin{equation}
M_n=n!\,[\log(n+e)]^{qn},\qquad n\ge 0.
\label{eq:Mn_def}
\end{equation}
Constants denoted by \(C,R,c,C_0,C_1,\ldots\) may change from line to line.  They may depend on \(q\) and on the fixed clock, but never on \(\nu,\Delta,A,m,\eps\), or \(\delta\).  All operator norms are spectral norms, and \(\|X\|_1\) denotes the integral of \(\|X(z)\|\) over the relevant \(z\)-interval. We use the usual big-O convention with  $\widetilde{O}(f(n))$ suppressing polylogarithmic factors in $f(n)$. 

\section{The flat Denjoy--Carleman clock}
\manualref{A.2}\label{app:dc-clock}

We first record two elementary estimates for the weights \(M_n\), and then construct the endpoint-flat clock.  The following construction is known as the infinite-convolution construction of a non-quasianalytic Denjoy--Carleman cutoff; see, e.g., Refs.~\cite{KrieglMichorRainer2009,Rainer2022}.

\begin{lemma}[Product stability]
\label{lem:shifted_product}
There is a constant \(C_M\ge1\), depending only on \(q\), such that for all nonnegative integers \(a,b,r\) and \(0\le \ell\le r\),
\begin{equation}
\binom r\ell M_{a+\ell}M_{b+r-\ell}
\le
C_M^{a+b+r+1}M_{a+b+r}.
\label{eq:shifted_product}
\end{equation}
More generally, if \(a_j,r_j\ge0\), \(\sum_{j=1}^p a_j=a\), and \(\sum_{j=1}^p r_j=r\), then
\begin{equation}
\binom{r}{r_1,\ldots,r_p}
\prod_{j=1}^p M_{a_j+r_j}
\le
C_M^{a+r+p}M_{a+r}.
\label{eq:multishifted_product}
\end{equation}
After increasing \(C_M\), one also has
\begin{equation}
M_{N+1}
\le
C_M^N\bigl(N[\log(N+e)]^q\bigr)^N,
\qquad N\ge1.
\label{eq:MNplusone_bound}
\end{equation}
\end{lemma}

\begin{proof}
For the logarithmic factors we use
\[
\log(a_j+r_j+e)\le \log(a+r+e).
\]
For the factorials,
\[
\binom{r}{r_1,\ldots,r_p}\prod_{j=1}^p(a_j+r_j)!
=
r!\prod_{j=1}^p\frac{(a_j+r_j)!}{r_j!}.
\]
Since
\[
\prod_{j=1}^p\frac{(a_j+r_j)!}{r_j!}
\le
(a+r)^a
\le
e^a\frac{(a+r)!}{r!},
\]
we obtain the multinomial estimate up to an exponential factor in \(a+r+p\).  The two-factor estimate is the case \(p=2\).  Finally, Stirling's bound and the elementary absorption of the extra factor \((N+1)[\log(N+e)]^q\) into \(C^N\) give \eqref{eq:MNplusone_bound}.
\end{proof}

\begin{lemma}[Flat Denjoy--Carleman switch]
\label{lem:dc_switch}
For every \(q>1\), there exists a monotone \(F_q\in C^\infty([0,1])\) such that
\begin{equation}
F_q(u)=0\quad(u\le1/4),\qquad
F_q(u)=1\quad(u\ge3/4),\qquad
F_q'(u)\ge0,
\label{eq:F_plateaus_app}
\end{equation}
and constants \(C_F,R_F>0\) such that
\begin{equation}
\|F_q^{(n)}\|_\infty
\le
C_FR_F^nM_n,
\qquad n\ge1.
\label{eq:F_derivative_app}
\end{equation}
For \(f_q=F_q'\), there are constants \(C_f,R_f>0\) such that
\begin{equation}
\|f_q^{(n)}\|_\infty
\le
C_fR_f^nM_n,
\qquad n\ge0.
\label{eq:f_derivative_app}
\end{equation}
\end{lemma}

\begin{proof}
Choose
\begin{equation}
a_j=\frac{a}{(j+2)[\log(j+e)]^q},
\qquad j\ge0,
\end{equation}
with \(a>0\) small enough that \(\sum_{j\ge0}a_j\le1/2\).  This is possible because \(q>1\).  Let
\[
\chi_j=a_j^{-1}{\bf 1}_{[-a_j/2,a_j/2]},
\qquad
\rho_N=\chi_0*\cdots*\chi_N .
\]
The sequence \(\rho_N\) converges, with all derivatives, to a compactly supported \(C^\infty\) probability density \(\rho\), supported in an interval of length at most \(1/2\).  For fixed \(n\) and \(N\ge n\), differentiating in the sense of distributions and placing the \(n\) derivatives on the first \(n\) box functions gives
\[
\rho_N^{(n)}
=
\chi_0'*\cdots*\chi_{n-1}'*\chi_n*\cdots*\chi_N .
\]
Since \(\|\chi_j'\|_{\mathrm{TV}}=2/a_j\), \(\|\chi_n\|_\infty=1/a_n\), and convolution with probability densities does not increase the \(L^\infty\) norm,
\begin{equation}
\|\rho_N^{(n)}\|_\infty
\le
\frac{2^n}{a_0a_1\cdots a_n}.
\end{equation}
Passing to the \(C^n\)-limit yields the same estimate for \(\rho\).  From the definition of \(a_j\),
\[
\frac{2^n}{a_0a_1\cdots a_n}
\le
CR^n(n+1)![\log(n+e)]^{q(n+1)}
\le
C'(R')^nM_n,
\]
where the final extra factor is absorbed into the exponential constant.

Define
\begin{equation}
F_q(u)=\int_{-\infty}^{u-1/2}\rho(x)\,dx .
\end{equation}
The support of \(\rho\) gives the plateaus in \eqref{eq:F_plateaus_app}, and monotonicity follows from \(\rho\ge0\).  Since \(F_q^{(n)}=\rho^{(n-1)}(\,\cdot-1/2)\) for \(n\ge1\), \eqref{eq:F_derivative_app} follows.  Finally, \(f_q^{(n)}=F_q^{(n+1)}\), and \(M_{n+1}\le C^nM_n\) after increasing the exponential constant, giving \eqref{eq:f_derivative_app}.
\end{proof}

\section{A weighted superadiabatic estimate}
\manualref{A.3}\label{app:weighted-superadiabatic}

The rescaled gap-time interval used below may have length of order \(\nu^{-1}\).  The following product class and superadiabatic theorem are arranged so that no constant grows with that interval length.

\begin{lemma}[Weighted \(L^\infty/L^1\) product calculus]
\label{lem:weighted_product_class}
Let \(I=[0,Z]\).  A smooth matrix-valued function \(X:I\to\mathbb C^{2\times2}\) has weighted order \(\mu\ge0\), with constants \(B,R\), if
\begin{equation}
\|X^{(r)}\|_\infty\le BR^rM_{\mu+r}\qquad(r\ge0),
\label{eq:weighted_order_infty}
\end{equation}
and, whenever \(\mu+r\ge1\),
\begin{equation}
\|X^{(r)}\|_1\le BR^rM_{\mu+r}.
\label{eq:weighted_order_one}
\end{equation}
No \(L^1\) bound is required in the single case \(\mu=r=0\).

If \(X\) has weighted order \(\mu\) and \(Y\) has weighted order \(\eta\), then \(XY\) has weighted order \(\mu+\eta\), with constants depending only on the previous constants and on \(q\).  The same holds for any fixed finite product.
\end{lemma}

\begin{proof}
Leibniz's rule gives
\[
(XY)^{(r)}=\sum_{\ell=0}^r\binom r\ell X^{(\ell)}Y^{(r-\ell)}.
\]
The \(L^\infty\) estimate follows from Lemma~\ref{lem:shifted_product}.  For the \(L^1\) estimate, assume \(\mu+\eta+r\ge1\).  In each Leibniz term at least one of \(\mu+\ell\) and \(\eta+r-\ell\) is positive; put the \(L^1\) norm on such a factor and the \(L^\infty\) norm on the other.  Lemma~\ref{lem:shifted_product} again controls the weights.  Iteration gives the finite-product claim.
\end{proof}

\begin{theorem}[Weighted global superadiabatic estimate]
\label{thm:weighted_global}
Let \(K:[0,Z]\to\mathbb C^{2\times2}\) be a smooth self-adjoint Hamiltonian with two simple eigenvalues separated by a gap at least \(2\).  Let \(P(z)\) be the instantaneous ground-state projector and \(Q(z)=\id-P(z)\).  Assume
\begin{equation}
K^{(n)}(0)=K^{(n)}(Z)=0\qquad(n\ge1),
\label{eq:K_flat_app}
\end{equation}
and, for constants \(C_0,R_0\),
\begin{equation}
\|K^{(n)}\|_\infty\le C_0R_0^nM_n,
\qquad
\|K^{(n)}\|_1\le C_0R_0^nM_n
\qquad(n\ge1).
\label{eq:weighted_deriv_assumption}
\end{equation}
Let \(U_A(z)\) solve
\begin{equation}
i\frac{d}{dz}U_A(z)=A K(z)U_A(z),
\qquad U_A(0)=\id .
\end{equation}
Then there exist constants \(C,c>0\), depending only on \(C_0,R_0,q\), such that for all \(A\ge1\),
\begin{equation}
\|(\id-P(Z))U_A(Z)P(0)\|
\le
C\exp\!\left[-c\frac{A}{[\log(e+A)]^q}\right].
\label{eq:weighted_global_bound}
\end{equation}
\end{theorem}

\begin{proof}
For an operator \(X\), write
\[
X^{\rm D}=PXP+QXQ,\qquad
X^{\rm OD}=PXQ+QXP.
\]
On off-diagonal operators, \(\operatorname{ad}_K:X\mapsto[K,X]\) is invertible with norm at most \(1/2\): if \(K=E_gP+E_eQ\) and \(E_e-E_g\ge2\), then
\[
[K,PXQ]=-(E_e-E_g)PXQ,\qquad
[K,QXP]=(E_e-E_g)QXP.
\]

\paragraph{Bounds for \(P\).}
We first show that \(P\) has weighted order \(0\):
\begin{equation}
\|P^{(n)}\|_\infty+\|P^{(n)}\|_1
\le
CR^nM_n,\qquad n\ge1,
\label{eq:P_weighted_bounds}
\end{equation}
and \(\|P\|_\infty=1\).  Differentiating \(P^2=P\) and \([K,P]=0\) \(n\) times gives
\begin{align}
PP^{(n)}+P^{(n)}P-P^{(n)}
&=
-\sum_{k=1}^{n-1}\binom nk P^{(k)}P^{(n-k)},
\label{eq:Pn_diag_weighted}
\\
[K,P^{(n)}]
&=
-\sum_{k=1}^{n}\binom nk[K^{(k)},P^{(n-k)}].
\label{eq:Pn_offdiag_weighted}
\end{align}
The left side of \eqref{eq:Pn_diag_weighted} determines the diagonal part of \(P^{(n)}\), while \eqref{eq:Pn_offdiag_weighted} and the gap determine the off-diagonal part.  Induction on \(n\), using Lemma~\ref{lem:weighted_product_class}, proves \eqref{eq:P_weighted_bounds}; in every \(L^1\) estimate, at least one factor has positive weighted index, so \(\|P\|_1\) is never used.

The same equations imply endpoint flatness.  If \(P^{(k)}=0\) at an endpoint for \(1\le k<n\), then \eqref{eq:K_flat_app} forces both the diagonal and off-diagonal parts of \(P^{(n)}\) to vanish there.  Hence
\begin{equation}
P^{(n)}(0)=P^{(n)}(Z)=0,\qquad n\ge1.
\label{eq:P_endpoint_flat}
\end{equation}

\paragraph{The homological inverse.}
Let \(Y=Y^{\rm OD}\) have weighted order \(\mu\ge1\), and let \(X=\mathfrak S(Y)\) be the unique solution of
\begin{equation}
i[K,X]=-Y,\qquad X^{\rm D}=0.
\label{eq:homological_equation}
\end{equation}
Then \(X\) has the same weighted order:
\begin{equation}
\|X^{(r)}\|_\infty+\|X^{(r)}\|_1
\le
CR^rM_{\mu+r},
\qquad r\ge0.
\label{eq:hom_weighted_bounds}
\end{equation}
For \(r=0\), this is the inverse-gap estimate.  For \(r\ge1\), differentiating \eqref{eq:homological_equation} gives
\[
i[K,X^{(r)}]
=
-Y^{(r)}
-i\sum_{\ell=1}^{r}\binom r\ell [K^{(\ell)},X^{(r-\ell)}],
\]
which controls the off-diagonal part of \(X^{(r)}\).  The diagonal part is controlled by differentiating \(PXP+QXQ=0\).  All remaining terms contain either a positive derivative of \(P\), a positive derivative of \(K\), \(Y\), or a lower derivative of \(X\), so Lemma~\ref{lem:weighted_product_class} closes the induction without any factor of \(Z\).  If all derivatives of \(Y\) vanish at the endpoints, then the same differentiated equations imply that all derivatives of \(X\) vanish at the endpoints.

\paragraph{Superadiabatic coefficients.}
Set \(\Pi_0=P\) and define
\[
D_A(X)=A^{-1}X'+i[K,X].
\]
We construct \(\Pi_n\) recursively.  Suppose \(\Pi_0,\ldots,\Pi_{n-1}\) have been chosen so that, with
\[
\Sigma_{n-1}=\sum_{j=0}^{n-1}A^{-j}\Pi_j,
\]
the projection defect \(\Sigma_{n-1}^2-\Sigma_{n-1}\) and invariance defect \(D_A(\Sigma_{n-1})\) vanish through order \(A^{-(n-1)}\).  Let \(G_n\) and \(F_n\) be the leading coefficients at order \(A^{-n}\):
\begin{equation}
G_n=\sum_{a=1}^{n-1}\Pi_a\Pi_{n-a},
\qquad
F_n=\Pi_{n-1}',
\label{eq:G_F_explicit}
\end{equation}
with the empty sum \(G_1=0\).

The coefficient \(G_n\) is diagonal.  Indeed, \(\Sigma_{n-1}^2-\Sigma_{n-1}\) commutes with \(\Sigma_{n-1}\), and the coefficient of \(A^{-n}\) in this commutator is \([P,G_n]\).  The coefficient \(F_n\) is off-diagonal.  Since \(D_A\) is a derivation,
\[
D_A(\Sigma_{n-1}^2-\Sigma_{n-1})
=
D_A(\Sigma_{n-1})\Sigma_{n-1}
+\Sigma_{n-1}D_A(\Sigma_{n-1})
-D_A(\Sigma_{n-1}).
\]
Taking the coefficient of \(A^{-n}\) gives
\[
i[K,G_n]=F_nP+PF_n-F_n.
\]
Because \(G_n\) is diagonal, it commutes with \(K\), and therefore \(F_nP+PF_n-F_n=0\), which is equivalent to \(F_n=F_n^{\rm OD}\).

Define
\begin{equation}
D_n=-PG_nP+QG_nQ,\qquad
O_n=\mathfrak S(F_n),\qquad
\Pi_n=D_n+O_n .
\label{eq:Pi_n_def}
\end{equation}
Then
\[
P\Pi_n+\Pi_nP-\Pi_n=-G_n,
\qquad
F_n+i[K,\Pi_n]=0,
\]
so both defects are cancelled through order \(A^{-n}\).

We now record the bounds.  The induction gives
\begin{equation}
\|\Pi_j^{(r)}\|_\infty+\|\Pi_j^{(r)}\|_1
\le
C^{j+r+1}R^{j+r}M_{j+r},
\qquad j+r\ge1.
\label{eq:Pij_weighted_bounds}
\end{equation}
For \(j=0\), this is \eqref{eq:P_weighted_bounds}.  If the estimate holds up to \(j=n-1\), then \(G_n\) has weighted order \(n\) by Lemma~\ref{lem:weighted_product_class}, \(F_n=\Pi_{n-1}'\) has weighted order \(n\), \(D_n\) has weighted order \(n\), and the homological inverse gives the same order for \(O_n\).  Finite sums and the number of Leibniz terms are absorbed into the exponential constants.

Endpoint flatness is inherited by the same recursion.  For \(j=0\), it is \eqref{eq:P_endpoint_flat}.  For \(j\ge1\), the terms defining \(G_j\) contain lower positive-order coefficients, and \(F_j=\Pi_{j-1}'\).  By induction, these are flat at the endpoints; hence so are \(D_j\), \(O_j\), and \(\Pi_j\).  Therefore
\begin{equation}
\Pi_j^{(r)}(0)=\Pi_j^{(r)}(Z)=0
\qquad(j\ge1,\ r\ge0).
\label{eq:Pi_endpoint_flat}
\end{equation}

For \(N\ge1\), set
\[
\Pi^{[N]}=\sum_{j=0}^{N}A^{-j}\Pi_j .
\]
By \eqref{eq:Pi_endpoint_flat},
\begin{equation}
\Pi^{[N]}(0)=P(0),\qquad \Pi^{[N]}(Z)=P(Z).
\label{eq:Pi_endpoint_values}
\end{equation}
The construction also gives the exact residual identity
\begin{equation}
\frac{d}{dz}\Pi^{[N]}+iA[K,\Pi^{[N]}]
=
A^{-N}\Pi_N'.
\label{eq:exact_residual_identity}
\end{equation}
Thus, by \eqref{eq:Pij_weighted_bounds} and \eqref{eq:MNplusone_bound},
\begin{align}
\|A^{-N}\Pi_N'\|_1
&\le
A^{-N}C^{N+2}R^{N+1}M_{N+1}
\nonumber\\
&\le
C\left(\frac{C_1N[\log(N+e)]^q}{A}\right)^N .
\label{eq:RN_bound}
\end{align}

Finally,
\[
\frac{d}{dz}\left(U_A(z)^\dagger\Pi^{[N]}(z)U_A(z)\right)
=
U_A(z)^\dagger
\left(A^{-N}\Pi_N'(z)\right)
U_A(z).
\]
Using \eqref{eq:Pi_endpoint_values} and integrating,
\[
\|P(Z)U_A(Z)-U_A(Z)P(0)\|
\le
\|A^{-N}\Pi_N'\|_1.
\]
Multiplying by \(\id-P(Z)\) on the left and \(P(0)\) on the right gives the same bound for the leakage.

Choose
\[
N=\left\lfloor a\frac{A}{[\log(e+A)]^q}\right\rfloor
\]
with \(a>0\) small enough.  For \(A\) larger than a \(q\)-dependent constant, \eqref{eq:RN_bound} is bounded by
\[
C\exp\!\left[-c\frac{A}{[\log(e+A)]^q}\right].
\]
For bounded \(A\), the leakage norm is at most one, and increasing \(C\) extends the estimate to every \(A\ge1\).
\end{proof}

\section{Proof of Theorem~\ref{thm:main_leakage}: resolved two-level leakage}
\manualref{A.4}\label{app:proof1}\label{app:proof-thm1}

We now apply Theorem~\ref{thm:weighted_global} to the two-level comparator
\[
h_\Delta(s)=-(1-s)X+s\Delta Z
\]
under the schedule \eqref{eq:single_schedule}.  Fix a resolved detuning \(1\ge|\Delta|>\nu\), and write
\[
d=|\Delta|,\qquad \sigma=\operatorname{sgn}(\Delta),\qquad
r_d(s)=\sqrt{(1-s)^2+d^2s^2}.
\]
Define the rescaled gap time
\begin{equation}
z(t)=\frac1A\int_0^t r_d(s(t'))\,dt'.
\label{eq:z_def}
\end{equation}
Then the two-level Schrödinger equation is
\begin{equation}
i\frac{d}{dz}U(z)=AK_\Delta(z)U(z),
\qquad
K_\Delta(z)=\frac{h_\Delta(s(t(z)))}{r_d(s(t(z)))}.
\label{eq:K_delta_def}
\end{equation}
The normalized Hamiltonian \(K_\Delta\) has eigenvalues \(\pm1\), hence gap \(2\), and has the same instantaneous projectors as \(h_\Delta\).

Introduce \(\phi\in[0,\pi/2]\) by
\begin{equation}
\cos\phi=\frac{1-s}{r_d(s)},
\qquad
\sin\phi=\frac{ds}{r_d(s)}.
\label{eq:phi_def_app}
\end{equation}
Then
\begin{equation}
K_\Delta=-\cos\phi\,X+\sigma\sin\phi\,Z\equiv G_\sigma(\phi).
\label{eq:K_phi_app}
\end{equation}
Let
\[
x=\frac{t}{T_{\rm stage}(A)},\qquad f(x)=F_q'(x).
\]
Using \eqref{eq:single_velocity}, \eqref{eq:single_stage_time}, and \eqref{eq:z_def}, one obtains
\begin{equation}
\frac{d}{dz}=\mathcal L,
\qquad
\mathcal L=a_d(\phi)\partial_x+b_{d,\nu}(\phi)f(x)\partial_\phi,
\label{eq:L_operator_app}
\end{equation}
where
\begin{align}
a_d(\phi)
&=
\frac{8}{\pi}\left(\nu\cos\phi+\frac{\nu}{d}\sin\phi\right),
\label{eq:a_phi_app}
\\
b_{d,\nu}(\phi)
&=
4(d\cos\phi+\sin\phi)
\left(\cos^2\phi+\frac{\nu^2}{d^2}\sin^2\phi\right).
\label{eq:b_phi_app}
\end{align}
The inequalities \(d\in[\nu,1]\) and \(\nu\in(0,1/2]\) imply uniform analytic bounds
\begin{equation}
\|\partial_\phi^r a_d\|_\infty+
\|\partial_\phi^r b_{d,\nu}\|_\infty+
\|\partial_\phi^r G_\sigma\|_\infty
\le
CR^r r!
\le
CR^rM_r.
\label{eq:analytic_coeff_bounds}
\end{equation}
Also \(a_d(\phi)>0\).

\begin{lemma}[Uniform derivative bounds for the normalized comparator]
\label{lem:K_weighted_bounds}
There are constants \(C,R>0\), depending only on \(q\) and on the fixed clock \(F_q\), such that for every \(\nu\in(0,1/2]\), every resolved \(d\in[\nu,1]\), and every \(n\ge1\),
\begin{equation}
\|K_\Delta^{(n)}\|_\infty\le CR^nM_n,
\qquad
\|K_\Delta^{(n)}\|_1\le CR^nM_n,
\label{eq:K_deriv_bounds_app}
\end{equation}
where derivatives and the \(L^1\) norm are taken with respect to \(z\).  Moreover,
\begin{equation}
K_\Delta^{(n)}(0)=K_\Delta^{(n)}(Z_\Delta)=0,
\qquad n\ge1,
\label{eq:K_endpoint_flat_app}
\end{equation}
where \(Z_\Delta\) is the final value of \(z(t)\).
\end{lemma}

\begin{proof}
Write
\[
f(x)=F_q'(x),
\qquad
\mathcal L
=
a_d(\phi)\partial_x+b_{d,\nu}(\phi)f(x)\partial_\phi,
\]
and define
\[
W_n=\mathcal L^nG_\sigma
\]
as a function of the independent variables
\((x,\phi)\in[0,1]\times[0,\pi/2]\).  Since
\(\mathcal L\) is the \(z\)-derivative along the trajectory, we have
\[
K_\Delta^{(n)}(z)=W_n(x(z),\phi(z)).
\]

We first prove the mixed derivative estimate
\begin{equation}
\|\partial_x^r\partial_\phi^s W_n\|_\infty
\le
C^{n+r+s+1}R^{n+r+s}M_{n+r+s}
\label{eq:W_mixed_bound}
\end{equation}
for all \(n,r,s\ge0\), uniformly in
\(\nu\in(0,1/2]\), \(d\in[\nu,1]\), and
\(\sigma\in\{\pm1\}\).  Here the supremum is taken over
\((x,\phi)\in[0,1]\times[0,\pi/2]\).

The case \(n=0\) follows from the analyticity of
\[
G_\sigma(\phi)=-\cos\phi\,X+\sigma\sin\phi\,Z .
\]
Moreover, since \(0<\nu/d\le1\) and \(0<d\le1\), the coefficients
\(a_d\), \(b_{d,\nu}\), and \(G_\sigma\) satisfy uniform analytic bounds:
for suitable constants \(C_a,R_a\),
\begin{equation}
\|\partial_\phi^\ell a_d\|_\infty
+
\|\partial_\phi^\ell b_{d,\nu}\|_\infty
+
\|\partial_\phi^\ell G_\sigma\|_\infty
\le
C_aR_a^\ell \ell!
\le
C_aR_a^\ell M_\ell .
\label{eq:analytic_coeff_bounds_local}
\end{equation}
The clock derivative satisfies
\[
\|f^{(\ell)}\|_\infty\le C_fR_f^\ell M_\ell .
\]

Assume \eqref{eq:W_mixed_bound} holds for \(W_n\).  Since
\[
W_{n+1}
=
a_d(\phi)\partial_xW_n
+
b_{d,\nu}(\phi)f(x)\partial_\phi W_n,
\]
Leibniz's rule gives two types of terms.  For the first term,
\[
\partial_x^r\partial_\phi^s
\bigl(a_d(\phi)\partial_xW_n\bigr)
=
\sum_{\ell=0}^s
\binom{s}{\ell}
(\partial_\phi^\ell a_d)
\,
\partial_x^{r+1}\partial_\phi^{s-\ell}W_n .
\]
Using \eqref{eq:analytic_coeff_bounds_local}, the induction hypothesis,
and Lemma~\ref{lem:shifted_product}, this is bounded by
\[
C^{n+r+s+2}R^{n+r+s+1}M_{n+r+s+1}.
\]
For the second term,
\[
\partial_x^r\partial_\phi^s
\bigl(b_{d,\nu}(\phi)f(x)\partial_\phi W_n\bigr)
=
\sum_{\alpha=0}^r\sum_{\beta=0}^s
\binom r\alpha\binom s\beta
f^{(\alpha)}(x)
(\partial_\phi^\beta b_{d,\nu})(\phi)
\partial_x^{r-\alpha}
\partial_\phi^{s-\beta+1}W_n .
\]
The derivative bounds for \(f\), the analytic bounds for \(b_{d,\nu}\),
the induction hypothesis, and Lemma~\ref{lem:shifted_product} again give
\[
C^{n+r+s+2}R^{n+r+s+1}M_{n+r+s+1}.
\]
After increasing \(C\) and \(R\), the induction closes.  Taking
\(r=s=0\) in \eqref{eq:W_mixed_bound} gives
\[
\|K_\Delta^{(n)}\|_\infty
\le
CR^nM_n .
\]

We now prove the \(L^1\) estimate.  Parameterize the trajectory by
\(x\mapsto\phi(x)\).  Since \(F_q\) is monotone, \(s(x)\) is monotone,
and hence \(\phi(x)\) is monotone increasing from \(0\) to \(\pi/2\).  Thus
\begin{equation}
\int_0^1\phi'(x)\,dx=\frac{\pi}{2}.
\label{eq:phi_total_variation}
\end{equation}
Along the trajectory,
\[
\mathcal L
=
a_d(\phi(x))
\bigl(\partial_x+\phi'(x)\partial_\phi\bigr)
=
a_d(\phi(x))\frac{d}{dx},
\qquad
dz=\frac{dx}{a_d(\phi(x))}.
\]
Therefore, for \(n\ge1\),
\begin{align}
\|K_\Delta^{(n)}\|_1
&=
\int_0^{Z_\Delta}
\|W_n(x(z),\phi(z))\|\,dz
\nonumber\\
&=
\int_0^1
\left\|
a_d(\phi(x))
\frac{d}{dx}W_{n-1}(x,\phi(x))
\right\|
\frac{dx}{a_d(\phi(x))}
\nonumber\\
&=
\int_0^1
\left\|
\frac{d}{dx}W_{n-1}(x,\phi(x))
\right\|dx
\nonumber\\
&\le
\int_0^1
\|\partial_xW_{n-1}(x,\phi(x))\|\,dx
+
\int_0^1
\phi'(x)
\|\partial_\phi W_{n-1}(x,\phi(x))\|\,dx
\nonumber\\
&\le
\|\partial_xW_{n-1}\|_\infty
+
\frac{\pi}{2}\|\partial_\phi W_{n-1}\|_\infty
\nonumber\\
&\le
CR^nM_n .
\end{align}
In the last step we used \eqref{eq:W_mixed_bound} with
\((n-1,r,s)=(n-1,1,0)\) and \((n-1,0,1)\).  This estimate is uniform in
\(\nu\) and \(d\); in particular, the length \(Z_\Delta\) of the
\(z\)-interval never appears.

Finally, \(F_q\) is constant near \(x=0\) and near \(x=1\).  Hence
\(s(x)\), \(\phi(x)\), and therefore \(K_\Delta(z)\), are constant on
neighborhoods of both endpoints.  All positive \(z\)-derivatives vanish
there, so
\[
K_\Delta^{(n)}(0)=K_\Delta^{(n)}(Z_\Delta)=0,
\qquad n\ge1 .
\]
\end{proof}

\begin{proof}[Proof of Theorem~\ref{thm:main_leakage}]
Lemma~\ref{lem:K_weighted_bounds} verifies the hypotheses of Theorem~\ref{thm:weighted_global} for \(K_\Delta(z)\), uniformly over all \(\nu\in(0,1/2]\) and all resolved detunings \(1\ge|\Delta|>\nu\).  Since \(K_\Delta\) is obtained from \(h_\Delta\) by division by the positive scalar \(r_d(s(t(z)))\), their instantaneous spectral projectors agree.  Reparameterizing time does not change the physical unitary.  Therefore
\[
\bigl\|
(\id-P_\Delta(1))\mathcal V_{\Delta,A}^{\rm stage}P_\Delta(0)
\bigr\|
\le
C\exp\!\left[-c\frac{A}{[\log(e+A)]^q}\right],
\]
which is \eqref{eq:main_leakage_bound}.
\end{proof}

\section{Choosing the adiabatic parameter}
\manualref{A.5}\label{app:choose-A}

\begin{lemma}[Inverting the near-exponential bound]
\label{lem:invert_bound}
Fix \(q>1\) and \(c>0\).  There is a constant \(C_q\) such that, for every \(L\ge1\), the choice
\begin{equation}
A=C_qL[\log(e+L)]^q
\label{eq:A_inversion_app}
\end{equation}
implies
\begin{equation}
\exp\!\left[-c\frac{A}{[\log(e+A)]^q}\right]\le e^{-L}.
\label{eq:invert_result_app}
\end{equation}
\end{lemma}

\begin{proof}
For a parameter \(\kappa\ge1\), set \(A_\kappa=\kappa L[\log(e+L)]^q\).  For all \(L\ge1\),
\[
\log(e+A_\kappa)
\le
B_q(1+\log\kappa)\log(e+L)
\]
with a constant \(B_q\) independent of \(L\) and \(\kappa\).  Hence
\[
\frac{A_\kappa}{[\log(e+A_\kappa)]^q}
\ge
\frac{\kappa}{B_q^q(1+\log\kappa)^q}\,L .
\]
Since \(\kappa/(1+\log\kappa)^q\to\infty\), choose \(\kappa=C_q\) large enough that the prefactor is at least \(1/c\).
\end{proof}

Let \(C_*,c_*\) be the constants from Theorem~\ref{thm:main_leakage}.  Choose
\begin{equation}
L=1+\log\frac{C_0m}{\sqrt{\delta}},
\qquad
A=C_qL[\log(e+L)]^q,
\label{eq:A_choice_app}
\end{equation}
where \(C_q\) is chosen using Lemma~\ref{lem:invert_bound} with \(c=c_*\), and \(C_0\) is large enough to absorb the prefactor \(C_*\).  Then the resolved-branch leakage amplitude of one stage obeys
\begin{equation}
\beta_{\rm stage}
\le
\frac{\sqrt{\delta}}{m}.
\label{eq:stage_beta_app}
\end{equation}
The corresponding stage time is
\begin{equation}
T_{\rm stage}(A)
=
\frac{\pi A}{8\nu}
=
O\!\left(
\frac1\nu L[\log(e+L)]^q
\right).
\label{eq:T_stage_app_final}
\end{equation}

\section{Proof of Theorem~\ref{thm:correctness}: composition over QPE stages}
\manualref{A.6}\label{app:composition}\label{app:proof-thm2}

We prove the measurement-success guarantee first for one fixed eigenvalue \(\lambda\), and then sum over eigenspaces.

For a length-\(j\) prefix \(b\in\{0,1\}^j\), let
\[
\val_j(b)=\sum_{\ell=1}^j2^{-\ell}b_\ell,
\qquad
I_b=[\val_j(b),\val_j(b)+2^{-j}]
\]
be the associated dyadic interval.  We call \(b\) good for \(\lambda\) if
\begin{equation}
\dist(\lambda,I_b)\le\nu.
\label{eq:good_prefix_app}
\end{equation}
Let \(G_j^\lambda\) be the projector onto the span of good length-\(j\) prefixes, tensored with the identity on unused qubits.  The empty prefix is good, so \(G_0^\lambda=\id\).

Consider stage \(j\), conditioned on a good parent prefix \(b\in\{0,1\}^{j-1}\).  Its midpoint is
\[
\mu_j(b)=\val_{j-1}(b)+2^{-j}.
\]
If \(|\lambda-\mu_j(b)|\le\nu\), then both children are good, since the midpoint is an endpoint of both child intervals.  If \(|\lambda-\mu_j(b)|>\nu\), then the child on the same side of \(\mu_j(b)\) as \(\lambda\) is good.  The opposite child is exactly the excited final state of the two-level comparator \(h_\Delta\), with \(\Delta=\lambda-\mu_j(b)\), and therefore has transition amplitude at most \(\sqrt{\delta}/m\) by \eqref{eq:stage_beta_app}.

This blockwise statement gives an operator bound for the full stage.  Let \(U_j^\lambda\) be the stage-\(j\) unitary restricted to the eigenspace of \(H_S\) with eigenvalue \(\lambda\), and define
\begin{equation}
\Gamma_{j-1}^\lambda
=
G_{j-1}^\lambda\otimes\ket{+}\!\bra{+}_{a_j}\otimes\id_{a_{>j}}.
\end{equation}
The fresh qubit \(a_j\) has not appeared in earlier stages, so the good component before stage \(j\) lies in the range of \(\Gamma_{j-1}^\lambda\).  Since the prefix blocks are orthogonal and the unresolved blocks have both children good,
\begin{equation}
\|(\id-G_j^\lambda)U_j^\lambda\Gamma_{j-1}^\lambda\|
\le
\frac{\sqrt{\delta}}{m}.
\label{eq:stage_good_projector_bound}
\end{equation}

Let
\[
\psi_j^\lambda
=
U_j^\lambda\cdots U_1^\lambda\ket{+}^{\otimes m},
\qquad
e_j^\lambda
=
\|(\id-G_j^\lambda)\psi_j^\lambda\|.
\]
Decompose \(\psi_{j-1}^\lambda\) into its good and bad parts.  The good part is controlled by \eqref{eq:stage_good_projector_bound}; the bad part has norm \(e_{j-1}^\lambda\), and unitarity cannot increase it.  Hence
\[
e_j^\lambda
\le
e_{j-1}^\lambda+\frac{\sqrt{\delta}}{m}.
\]
Since \(e_0^\lambda=0\), iteration gives
\begin{equation}
e_m^\lambda\le\sqrt{\delta}.
\label{eq:em_bound_app}
\end{equation}

If a final prefix \(b\) is good, then \(\dist(\lambda,I_b)\le\nu\) and \(|I_b|=2^{-m}\).  Therefore the center estimate satisfies
\begin{equation}
|\widehat\lambda(b)-\lambda|
\le
\nu+\frac{2^{-m}}2
=
2^{-m}
\le
\eps,
\label{eq:center_estimate_success_app}
\end{equation}
using \(\nu=2^{-m-1}\) and \(m=\lceil\log_2(1/\eps)\rceil\).  Thus \(G_m^\lambda\) is supported on the valid-estimate projector \(G^\lambda\) from Eq.~\eqref{eq:fail}.

For a general input state, write
\[
H_S=\sum_\alpha\lambda_\alpha\Pi_\alpha,
\qquad
\ket{\psi}=\sum_\alpha\ket{\psi_\alpha},
\qquad
\ket{\psi_\alpha}=\Pi_\alpha\ket{\psi}.
\]
All stage Hamiltonians commute with \(H_S\), so the evolution is block diagonal in the eigenspaces of \(H_S\).  Consequently,
\begin{align}
\left\|
\Pi_{\rm fail}\,
\mathcal V_A
\left(
\ket{\psi}_{\mathcal S}\ket{+}_{\mathcal A}^{\otimes m}
\right)
\right\|^2
&=
\sum_\alpha
\|\ket{\psi_\alpha}\|^2
\left\|
(\id-G^{\lambda_\alpha})\psi_m^{\lambda_\alpha}
\right\|^2
\nonumber\\
&\le
\delta .
\end{align}
This proves the correctness claim of Theorem~\ref{thm:correctness}.

It remains to compute the total time.  By \eqref{eq:T_stage_app_final},
\[
T_{\rm tot}
=
mT_{\rm stage}
=
O\!\left(
\frac{m}{\nu}L[\log(e+L)]^q
\right),
\qquad
L=1+\log\frac{C_0m}{\sqrt{\delta}}.
\]
Since
\[
m=\Theta(\log(1/\eps)),
\qquad
\nu=\Theta(\eps),
\qquad
L=
O\!\left(
1+\log\frac1\delta+\log\log\frac e\eps
\right),
\]
this is exactly the runtime bound stated in Theorem~\ref{thm:correctness}.
\section{Application to local adiabatic Grover search}
\manualref{A.7}\label{app:grover}

We briefly indicate why the same argument gives the high-confidence Grover statement in the main text.  Let \(M\) marked items be known, set \(a=M/N\), and assume \(0<a\le1/2\).  The standard local adiabatic search Hamiltonian is
\begin{equation}
H_G(s)=(1-s)(\id-\ket{u}\!\bra{u})+s(\id-\Pi_M),
\label{eq:grover_H_app}
\end{equation}
where \(\ket{u}\) is the uniform state and \(\Pi_M\) projects onto the marked subspace \cite{RolandCerf2002}.  Starting from \(\ket{u}\), the evolution is exactly confined to the two-dimensional subspace spanned by
\begin{equation}
\ket{w}=\frac{\Pi_M\ket{u}}{\sqrt a},
\qquad
\ket{r}=\frac{(\id-\Pi_M)\ket{u}}{\sqrt{1-a}} .
\end{equation}
In this subspace, after subtracting the scalar \(\id/2\) and applying an \(s\)-independent rotation of Pauli axes, the traceless Hamiltonian is
\begin{equation}
\overline H_a(s)
=
-\frac12\left[
2\sqrt{1-a}\left(s-\frac12\right)X+\sqrt a\,Z
\right],
\label{eq:grover_affine_crossing_app}
\end{equation}
with gap
\begin{equation}
g_a(s)=
\sqrt{a+4(1-a)\left(s-\frac12\right)^2}.
\label{eq:grover_gap_app}
\end{equation}
Thus local adiabatic Grover search is another affine two-level avoided crossing.

Define the Roland--Cerf coordinate
\begin{equation}
\Theta_a(s)=\int_0^s\frac{du}{g_a(u)^2}.
\label{eq:grover_theta_app}
\end{equation}
A direct integration gives
\begin{equation}
\Theta_a(1)
=
\frac{\arctan\sqrt{(1-a)/a}}{\sqrt{a(1-a)}}
=
\Theta\!\left(\sqrt{\frac{N}{M}}\right).
\label{eq:grover_theta_length_app}
\end{equation}
Choose the DC-clocked local schedule by
\begin{equation}
\Theta_a(s(t))=\Theta_a(1)F_q(t/T),
\qquad
T=A\Theta_a(1),
\label{eq:grover_dc_schedule_app}
\end{equation}
up to an immaterial constant rescaling of \(A\).  To verify the hypotheses of Theorem~\ref{thm:weighted_global}, introduce the Bloch angle
\begin{equation}
\theta(s)=
\arctan\!\left(
\frac{2\sqrt{1-a}(s-\frac12)}{\sqrt a}
\right),
\qquad
\theta_0=\arctan\sqrt{\frac{1-a}{a}} .
\end{equation}
Then
\begin{equation}
d\Theta_a=\frac{d\theta}{2\sqrt{a(1-a)}},
\qquad
\theta(t)=-\theta_0+2\theta_0F_q(t/T),
\label{eq:grover_angle_clock_app}
\end{equation}
and the normalized traceless Hamiltonian is a fixed-gap two-level path
\begin{equation}
K_a(\theta)
=
-\sin\theta\,X-\cos\theta\,Z
\end{equation}
with eigenvalues \(\pm1\).  In the corresponding gap-time variable, \(d/dz\) has the form
\begin{equation}
\frac{d}{dz}
=
c_a(\theta)\partial_x
+
2\theta_0c_a(\theta)F_q'(x)\partial_\theta,
\qquad
x=t/T,
\label{eq:grover_L_app}
\end{equation}
where
\begin{equation}
c_a(\theta)=\frac{2\sqrt{1-a}\cos\theta}{\theta_0}
\end{equation}
after absorbing a constant into \(A\).  Since \(a\le1/2\), we have
\(\theta_0\in[\pi/4,\pi/2]\), so \(c_a\), its \(\theta\)-derivatives, and the coefficient \(2\theta_0c_a\) are uniformly bounded.  The same induction as in Lemma~\ref{lem:K_weighted_bounds} therefore gives
\begin{equation}
\|K_a^{(n)}\|_\infty+\|K_a^{(n)}\|_1
\le
CR^nM_n,
\qquad n\ge1,
\label{eq:grover_derivative_bounds_app}
\end{equation}
with constants independent of \(N\) and \(M\).  The \(L^1\) estimate again uses only the total angle variation
\begin{equation}
\int_0^1 |\theta'(x)|\,dx=2\theta_0\le\pi,
\end{equation}
and hence does not pay for the long Roland--Cerf coordinate length
\(\Theta_a(1)=\Theta(\sqrt{N/M})\).  Endpoint flatness follows from the fact that \(F_q\) is constant near \(x=0,1\).

Applying Theorem~\ref{thm:weighted_global} gives leakage amplitude
\begin{equation}
\eta_G(A)
\le
C\exp\!\left[
-c\frac{A}{[\log(e+A)]^q}
\right].
\end{equation}
At \(s=1\), the ground state in the invariant two-dimensional subspace is \(\ket{w}\), which lies in the marked subspace.  Hence the probability of measuring an unmarked item is bounded, after adjusting constants, by
\begin{equation}
p_{\rm fail}
\le
C\exp\!\left[
-c\frac{A}{[\log(e+A)]^q}
\right],
\qquad
T=A\Theta_a(1)
=
\Theta\!\left(A\sqrt{\frac{N}{M}}\right).
\label{eq:grover_failure_app}
\end{equation}
Thus the DC-clocked Roland--Cerf schedule gives a single-run, high-confidence adiabatic Grover evolution with near-exponential failure suppression in the dimensionless adiabatic parameter \(A\), assuming the marked fraction \(M/N\) is known.

\end{document}